\documentclass{article}

\usepackage{amsmath, amsthm, amssymb, dsfont, mathtools}
\usepackage{mathrsfs}

\usepackage{graphicx}
\usepackage{natbib}
\usepackage{caption}
\usepackage{enumerate}
\usepackage[inline]{enumitem}
\usepackage{hyperref}
\usepackage[margin=1.5in]{geometry}
\hypersetup{colorlinks,citecolor=blue,urlcolor=blue,linkcolor=blue}
\usepackage{nikos_tex}
\usepackage{float}

\usepackage{multirow}

\usepackage[utf8]{inputenc}
\usepackage[english]{babel}

\graphicspath{{./figures/}}

\theoremstyle{definition}
\newtheorem{prop}{Proposition}

\newtheorem{lemm}[prop]{Lemma}
\newtheorem{theo}[prop]{Theorem}
\newtheorem{rema}[prop]{Remark}

\newtheorem{assu}[prop]{Assumption}

\theoremstyle{remark}

 \usepackage{epigraph}

 \usepackage{etoolbox}
 
 \makeatletter
 \patchcmd{\epigraph}{\@epitext{#1}}{\itshape\@epitext{#1}}{}{}
 \makeatother

\date{Draft manuscript, December 2025}

\title{Partially Bayes p-values for large scale inference}

\author{
Nikolaos Ignatiadis\\
 \texttt{ignat@uchicago.edu}
 \and
 Li Ma \\
 \texttt{li.ma@uchicago.edu}
}

\begin{document}

\maketitle

\begin{abstract}
We seek to conduct statistical inference for a large collection of primary parameters, each with its own nuisance parameters. Our approach is partially Bayesian, in that we treat the primary parameters as fixed while we model the nuisance parameters as random and drawn from an unknown distribution which we endow with a nonparametric prior. We compute partially Bayes p-values by conditioning on nuisance parameter statistics, that is, statistics that are ancillary for the primary parameters and informative about the nuisance parameters. The proposed p-values have a Bayesian interpretation as tail areas computed with respect to the posterior distribution of the nuisance parameters. Similarly to the conditional predictive p-values of Bayarri and Berger, the partially Bayes p-values avoid double use of the data (unlike posterior predictive p-values). A key ingredient of our approach is that we model nuisance parameters hierarchically across problems; the sharing of information across problems leads to improved calibration.
We illustrate the proposed partially Bayes p-values in two applications: the normal means problem with unknown variances and a location-scale model with unknown distribution shape. We model the scales via Dirichlet processes in both examples and the distribution shape via {P\'olya} trees in the second. Our proposed partially Bayes p-values increase power
and calibration compared to purely frequentist alternatives.
\\


\end{abstract}

\section{Introduction} 
\label{sec:introduction}

We study the following common scenario in large-scale inference. We are faced with $n$ parallel statistical tasks pertaining to $n$ units of interest. For the $i$-th unit, we observe data $\dataset_i$ whose distribution is parameterized by a primary parameter $\primary_i \in \primaryspace$ and a nuisance parameter $\nuisance_i \in \nuisancespace$. Our goal is to conduct statistical inference for the primary parameters, $\primary_1,\dotsc,\primary_n$, by testing the null hypotheses $H_{1}: \primary_1=\primary_0, \dotsc, H_{n}: \primary_n = \primary_0$, where $\primary_0 \in \primaryspace$ is a pre-determined value. We seek to do so, by  effectively accounting for uncertainty in the nuisance parameters $\nuisance_1,\dotsc,\nuisance_n$ and by sharing information about the nuisance parameters across the $n$ units (but not the primary parameters). We achieve this goal by leveraging Bayesian nonparametrics and generalizing the conditional predictive p-values of~\citet{bayarri2000values} and the related secondarily Bayes p-values of~\citet{brown1965secondarily}.

As a starting point for our proposal, we summarize the $i$-th dataset $\dataset_i$  as  $\dataset_i \mapsto (T_i, U_i)$ for two statistics $T_i$ and $U_i$ (Fig.~\ref{fig:conditional_pred_pvalues}A). The statistic $T_i=T(\dataset_i) \in \RR$ is such that large values of $\abs{T_i}$ represent increasing evidence against the null $H_i: \primary_i = \primary_0$. The distribution of $T_i$ may depend on both $\primary_i$ and $\nuisance_i$. The statistic $U_i = U(\dataset_i) \in \nuisancestatisticspace$ is ancillary for the primary parameter $\primary_i$, that is, its distribution only depends on the nuisance parameter $\nuisance_i$. We call $T_i$ the test statistic and $U_i$ the nuisance parameter statistic. We then posit the following hierarchical model:
\begin{subequations}
\label{eq:hierarchical}
\begin{align} 
(T_i, U_i) \cond \primary_i, \nuisance_i \;\; &\simindep \;\; p(t, u \cond \primary_i, \nuisance_i ), \label{eq:hierarchy1}  \\ 
\nuisance_i \cond G \;\; & \simiid \;\;  G, \label{eq:hierarchy2} \\ 
G \;\; &\sim \;\; \Pi. \label{eq:hierarchy3}
\end{align}
\end{subequations}
\noindent Above, $p(t,u \mid \primary_i,\nuisance_i)$ denotes the density of $(T_i, U_i)$ given the unknown parameters $\primary_i,\nuisance_i$. 
Our approach to inference is partially Bayesian~\citep{cox1975note, mccullagh1990note}, since we treat the primary parameters $\primary_1,\dotsc,\primary_n$ as fixed while we model the nuisance parameters $\nuisance_1,\dots,\nuisance_n$ as random draws from a distribution $G$. The distribution $G$
is itself modeled as a draw from a prior $\Pi$, usually specified via a Bayesian nonparametric process to allow $G$ to take a variety of forms, although $\Pi$ may also be a parametric prior for some applications. Given the hierarchical model in~\eqref{eq:hierarchical}, we propose to compute partially Bayes (PB) p-values for all units $i=1,\dotsc,n$ by evaluating the null tail area of $T_i$ conditional on all nuisance parameter statistics, $U_1,\dotsc,U_n$,

\begin{equation}
\label{eq:conditional_pred_pvalue}
\begin{aligned}
&\cpvalue_i := \cpvaluefun_i(T_i, (U_1, \dotsc, U_n),\, \Pi),\;\;\\
&\cpvaluefun_i(t, (u_1,\dotsc,u_n),\, \Pi) := \Pi(\abs{T_i'} \geq \abs{t} \,\mid \, U_1 = u_1,\dotsc, U_n=u_n),
\end{aligned}
\end{equation}
where $T_i'$ is an identically distributed copy of $T_i$ under the null, $\primary_i = \primary_0$, conditional on $U_1,\dotsc,U_n$. For intuition, consider the case wherein in~\eqref{eq:hierarchy1}, $U_i$ is independent of $T_i$ conditional on $\primary_i, \nuisance_i$. Then, $T_i'$ may be generated as follows: sample $\nuisance_i'$ from the posterior distribution of $\nuisance_i$ conditional on \emph{all} the nuisance parameter statistics $U_1,\dotsc,U_n$ and then draw $T_i'$ from the conditional distribution of $T_i$ given $\primary_i=\primary_0$ and $\nuisance_i = \nuisance_i'$.

The computation of~\eqref{eq:conditional_pred_pvalue} may be conceptualized as consisting of two distinct steps illustrated in Fig.~\ref{fig:conditional_pred_pvalues}B,C: first, conduct Bayesian inference to compute the tail-area function of $T_i$ under the null, conditional on $U_1,\dotsc,U_n$; this step operationalizes sharing of information across units for the nuisance parameters. Second, use the derived tail area function in a frequentist fashion by evaluating it at the observed value of the test statistic $T_i$.  

\begin{figure}
\centering
\includegraphics[width=0.9\textwidth]{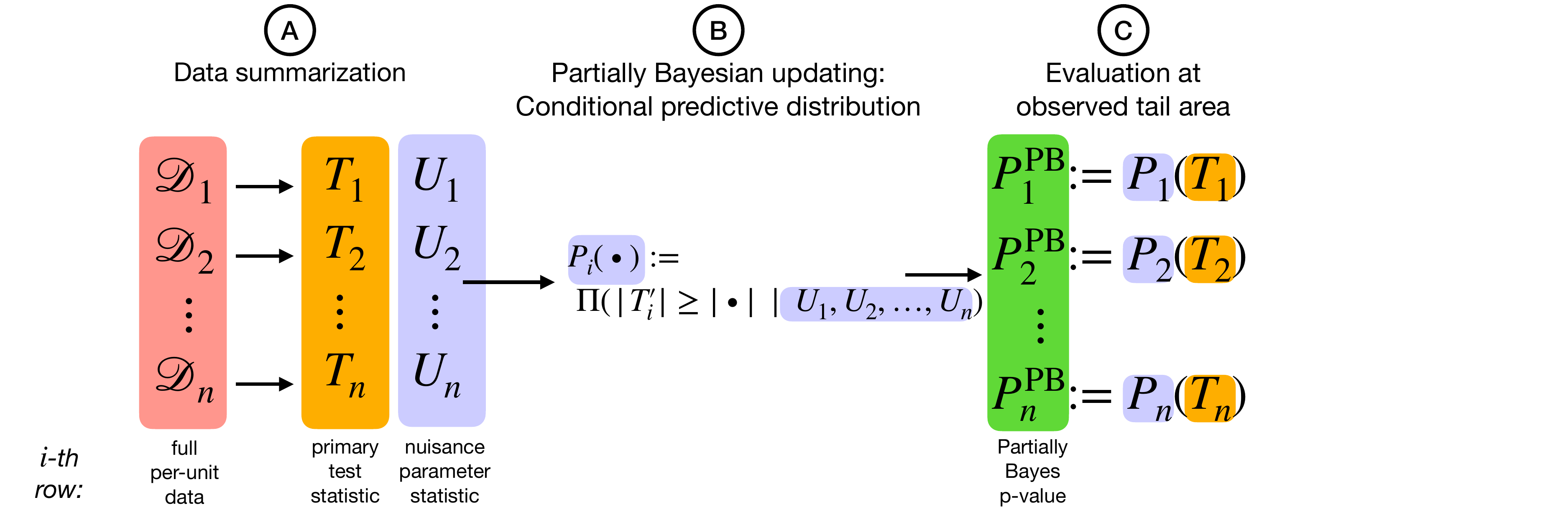}
\caption{Illustration of the three steps for computing partially Bayes p-values. (A) First, the data for each unit $i$ is summarized into a test statistic $T_i$ and a nuisance parameter statistic $U_i$. (B) Second, we conduct Bayesian inference, pooling information from all nuisance statistics $U_1, \dots, U_n$, to compute the tail-area function $\cpvaluefun_i(\cdot, (U_1, \dotsc, U_n),\, \Pi)$ in~\eqref{eq:conditional_pred_pvalue}. (C) Finally, the $i$-th partially Bayes p-value $\cpvalue_i$ is obtained by evaluating this function at the observed test statistic $T_i$.}
\label{fig:conditional_pred_pvalues}
\end{figure}

The proposed p-values are neither fully frequentist nor Bayesian; the proposal is a pragmatic and practical attempt at getting benefits of both worlds. From the frequentist perspective, our goal is to develop a new practical and general approach for constructing  powerful (partially Bayes) p-values that are approximately calibrated~\citep{rubin1984bayesianly, meng1994posterior, robins2000asymptotic} under the null in the presence of nuisance parameters. By calibration here we mean approximate uniformity, that is, for $i$ such that $\primary_i = \primary_0$, $\mathbb P[ \cpvalue_i \leq \alpha] \approx \alpha$ for all $\alpha \in [0,1]$. The calibration may hold conditional on nuisance parameters or not, and we defer a more detailed discussion to Section~\ref{sec:calibration}. What is important here, is that calibration will hold under two asymptotic regimes: as the per-unit sample size increases (that is, as we collect more data regarding each individual $\primary_i$), and also as the number of units $n$ increases with the per-unit sample size
remaining fixed. Existing approaches, e.g., the conditional predictive p-values of~\citet{bayarri2000values}, have guarantees under the former asymptotics~\citep{robins2000asymptotic}, but not under the latter.
From the Bayesian perspective, our goal is to gain power through principled information sharing across units by leveraging flexible data-adaptive modeling techniques developed in Bayesian nonparametrics~\citep{muller2015bayesian, ghosal2017fundamentals}.

\subsection{Motivation: Industrialist p-values}
\label{subsec:industrialist}
The motivation for this work stems from the widespread usage of p-values in high-throughput biological experiments as a convenient device for screening thousands of hypotheses; an intermediate step to be followed up by futher experimentation.
\citet{huber2016clash} calls such p-values ``industrialist,'' contrasting them to ``craftsperson'' p-values reported as the final statistical output of an analysis seeking to answer a single predefined question. Practitioners are used to reasoning about multiple testing procedures via p-values, about their approximate uniformity under the null, and to conducting visual inspections of p-value histograms and qq-plots, as explained in e.g.,~\citet{li2013statistical, robinson2014how}, \citet{ignatiadis2016datadriven},~\citet{breheny2018pvalue}, and~\citet[Chapter 6.9.1]{holmes2019modern}.

Our main thrust is that given their exploratory nature, industrialist p-values do not necessarily need to conform to the desiderata of purely frequentist p-values, that is, to require uniformity conditional on any value of the nuisance parameters. By permitting weaker notions of approximate uniformity, we can gain power and flexibility in large-scale inference, and may be able to construct p-values in settings wherein purely frequentist p-values are not available or would be overly conservative. 

As a case in point,
one of the standard ways of computing p-values in high-throughput biology, e.g., for microarray and RNA-Seq studies, is via the limma R package~(\citealp{smyth2004linear, ritchie2015limma}, $>$50,000 combined citations). 
We will explain further below that the way limma computes p-values  is conceptually very close to the partially Bayes p-values we propose; the main difference is that the information pooling in~\eqref{eq:hierarchical} is achieved via empirical Bayes rather than (nonparametric) hierarchical Bayes. Following~\citet{ignatiadis2024empirical}, we use the terminology ``empirical partially Bayes'' for methods such as limma.

For illustration, we consider the study of
\citet{palmieri2015genomewide} who collected samples from patients with Crohn's disease, and performed a genome-wide analysis comparing gene expression from inflamed vs.\ noninflamed colonic mucosa. Gene expression was measured with Affymetrix microarrays. We preprocessed the dataset following~\citet{klaus2018end}, leading to measurements for $n=16,125$ genes in $24$ samples (12 patients, inflamed and noninflamed sample per patient). A paired difference for the $i$-th gene
(pairing each patient's inflamed vs. noninflamed sample) yielded measurements $Z_{i1}, \dotsc, Z_{iK}$ with $K=12$, which we model as,
\begin{equation}
\label{eq:location_unknown_shape_intro}
\dataset_i = \cb{Z_{i1}, \dotsc, Z_{iK_i}},\;\;\; Z_{ij} \mid F_i  \simiid  F_i.
\end{equation}
The gene-specific distribution $F_i$ is modeled as $F_i(\cdot) = W_i(\cdot - \primary_i)$, where the primary parameter $\primary_i$ is the mean of $F_i$ and the nuisance parameter $\nuisance_i = W_i$ is a centered distribution with mean $0$. To screen for differentially expressed genes, we test the null hypotheses $H_i: \primary_i = 0$ for all $i=1,\dotsc,n$, which asks whether the expected pairwise difference in gene expression is zero.
We illustrate three strategies for computing p-values.
\begin{enumerate}[leftmargin=*]
\item \textbf{Standard t-test p-values:} We make the normality assumption that $W_i = \mathrm{N}(0, \sigma_i^2)$. Hence we are effectively testing $H_i: \theta_i=0$ based on \smash{$Z_{i1},\dotsc,Z_{iK} \simiid \mathrm{N}(\theta_i, \sigma_i^2)$}. We compute the p-value for $H_i$ using the standard two-sided t-test. The original study~\citep{palmieri2015genomewide} used t-tests, reporting genes with p-values $\leq 0.001$ as significant. 
\item \textbf{Partially Bayes p-values under normality:} We continue to make the same normality assumption as above. Following our partially Bayes framework, we treat $\theta_i$ as the primary parameter and \smash{$\sigma_i^2$} as the nuisance parameter. We let the primary test statistic $T_i$ be the sample average $\bar{Z}_i$ of the $Z_{ij}$, and use the sample variance $S_i^2$ of the $Z_{ij}$ as the nuisance parameter statistic $U_i$ (that is, \smash{$T_i = \bar{Z}_i$}, \smash{$U_i=S_i^2$}). We choose $\Pi$ in~\eqref{eq:hierarchy3} as a Dirichlet Process on $\mathbb R_+$. We postpone a more detailed description to Section~\ref{sec:normal_means} and for now only mention that the approach is very similar to limma with two modifications: a nonparametric model (Dirichlet Process) replaces a conjugate parametric prior and hierarchical Bayes replaces empirical Bayes. 
\item \textbf{Partially Bayes p-values with unknown distribution shape:} We now lean more closely to the general specification in~\eqref{eq:location_unknown_shape_intro} and pursue the following opportunity: instead of positing a normal noise model (as above), we seek to also learn the noise model from the data. To this end, we keep the same primary test statistic $T_i = \bar{Z}_i$ and let the nuisance parameter statistic $U_i$ be equal to the configuration $(Z_{i1}-T_i,\dotsc,Z_{iK}-T_i)$, whose distribution does not depend on $\primary_i$. The nuisance parameters now are given by the centered distributions $\nuisance_i = W_i$ and we model these as $W_i = W(\cdot / \tau_i)$, where $W$ is the same for all $i$ and drawn from a symmetrized P\'olya tree \citep{lavine1992aspects,walker1999bayesian}. Meanwhile, we model heteroscedasticity via the random scales $\tau_i$, which are distributed according to a distribution drawn from a Dirichlet Process. See Section~\ref{sec:unknown_shape} for details. 
\end{enumerate}

\begin{figure}
  \centering
  \begin{tabular}{c}
  \includegraphics[width=0.55\linewidth]{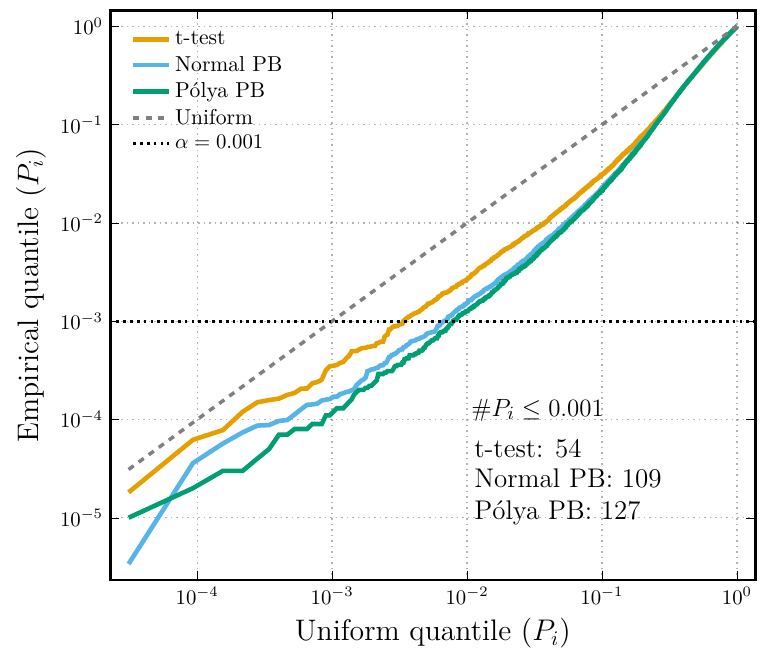}  
  \end{tabular}
  \caption{Reanalysis of the study of~\citet{palmieri2015genomewide}. The qq-plot compares the quantiles of the t-test p-values and the two types of partially Bayes p-values against the uniform quantiles. We observe that for large p-values all three methods are approximately uniform, but the left tail shows differences: the partially Bayes p-values are smaller than the t-test p-values, indicating higher power. The panel also indicates the number of p-values $\leq 0.001$ for each method.}
  \label{fig:palmieri}
\end{figure}

In Fig.~\ref{fig:palmieri} we show a qq-plot of all three p-values computed over all $n=16,125$ genes versus the uniform quantiles. We also report the number of p-values (from each method) that are $\leq 0.001$ (the threshold used in the original study).
We observe that both partially Bayes p-values are more powerful than the standard t-test p-values and that the partially Bayes p-values with unknown distribution shape are more powerful than the partially Bayes p-values under normality.

\subsection{Summary of contributions}
Our main contribution is the introduction of partially Bayes p-values, a new method for large-scale inference that handles nuisance parameters by pooling information using Bayesian nonparametric models.
In Section~\ref{sec:related_work}, we situate our proposal within the literature on conditional predictive p-values and existing partially Bayes methods.
Section~\ref{sec:calibration} provides a detailed theoretical analysis of the calibration properties of our proposed p-values under different sampling frames and asymptotic regimes. We establish asymptotic calibration as the per-unit sample size grows (Theorem~\ref{theo:K_to_infty}) and, more importantly for large-scale inference, as the number of units grows, both in the empirical Bayes (Theorem~\ref{theo:calibration_n_to_infty_empirical_bayes}) and frequentist (Theorem~\ref{theo:calibration_n_to_infty_frequentist}) frames.
In Section~\ref{sec:normal_means}, we apply our method to the normal means problem with unknown variances, using a Dirichlet Process prior; this corresponds to the ``Partially Bayes p-values under normality'' method in our introductory example, which shows substantial power gains (Figure~\ref{fig:palmieri}).
Section~\ref{sec:unknown_shape} extends the framework to the more challenging setting of location problems with unknown distributional shape, which we model using a P\'olya tree prior for the distribution shape and a Dirichlet Process for the noise scales. 
Section~\ref{sec:partially_bayes_summarization} introduces a practical method for visualizing the resulting decision boundaries. In Section~\ref{sec:numerical_study}, we validate our approach through simulation studies. Finally, in Section~\ref{sec:real_data_application}, we demonstrate the practical utility of our approach by revisiting the Crohn's disease gene expression study from Section~\ref{subsec:industrialist} and by analyzing a study on the health effects of low-frequency magnetic fields.

\section{Related work}
\label{sec:related_work}
To contextualize our proposal, we first provide a historical remark on the secondarily Bayes approach introduced in the PhD thesis of~\citet{brown1965secondarily} (also see~\citet{brown1967twomeans}) following a suggestion and supervision by John Tukey. To quote~\citet{brown1965secondarily}: ``Parameters in a problem may often be divided into two categories; those of direct interest (primary) and those required to estimate the precision of the estimators of the primary parameter. The latter type of parameter may be called secondary parameteters. The \emph{fully} Bayes approach to a problem requires that a prior distribution be assumed on \emph{all} parameters of interest in a problem; the \emph{secondarily} Bayes approach requires that a prior be assumed only on the \emph{secondary parameters}. The effect on the distribution of a test statistic caused by a secondary prior (on the secondary parameters) is normally less than would be caused by a similar appearing prior assumed on the primary parameters. The estimate of the primary parameter is usually obtained without regard to the secondary prior; only the estimate of the precision of this primary estimator will be affected.'' As one example,~\citet{brown1965secondarily} studies the normal means problem with unknown variance that we will consider in detail in Section~\ref{sec:normal_means}. Using our notation, the p-values of Morton Brown may be written as $P_i = \Pvalfunoracle(T_i, U_i, G)$, where (the abbreviation ``or'' refers to oracle and will be explained in Section~\ref{sec:calibration})
\begin{equation}
\Pvalfunoracle(t, u, G) := \PP[G]{\abs{T_i'} \geq \abs{t} \, \mid \, U_i=u} = \PP{\abs{T_i'} \geq \abs{t} \, \mid \, U_i=u, G}.
\label{eq:oracle_G_pvalue}
\end{equation}
In words, these p-values are analogous to our proposed partially Bayes p-values in~\eqref{eq:conditional_pred_pvalue}, but with the difference that they only account for the first two levels of the hierarchy in~\eqref{eq:hierarchical}, that is~\eqref{eq:hierarchy1} and~\eqref{eq:hierarchy2}, and require the analyst to specify a prior $G$ on the nuisance parameter rather than specifying a prior $\Pi$ on $G$. Thus, our proposed p-values are a direct extension of secondarily Bayes to large scale inference that allows for learning the nuisance parameter distribution $G$ from the hierarchical structure of the data. On the terminology front, we depart from the term ``secondarily Bayes'' and use the term ``partially Bayes'' as suggested by~\citet{cox1975note} and call the secondary parameters, nuisance parameters. 

From a different perspective, p-values akin to the ones in~\eqref{eq:oracle_G_pvalue} are proposed by~\citet{bayarri1999quantifying, bayarri2000values} under the name conditional predictive p-values. Compared to prior predictive p-values~\citep{box1980sampling} and posterior predictive p-values~\citep{guttman1967use, meng1994posterior}, these p-values enjoy several appealing theoretical properties: (i) for proper $G$, they arise naturally as the conditional distribution of the test statistic given the conditioning statistic under the prior predictive measure; (ii) their conditioning structure prevents double use of the data by separating the computation of the nuisance parameter posterior from the tail probability calculation (cf. Fig.\ref{fig:conditional_pred_pvalues}); (iii) this conditioning also ensures that the p-values primarily capture surprise in the data under the model, with the prior on nuisance parameters playing only a secondary role; and (iv) unlike prior predictive p-values, their construction remains valid even with improper prior $G$. While their approach treats each testing problem individually using non-informative priors, our work extends these ideas to the multiple testing setting by endowing $G$ with a nonparametric prior, allowing us to learn a proper prior $G$ by pooling information across the many hypotheses.

This work closely leans on the empirical partially Bayes p-values studied by~\citet{ignatiadis2024empirical}. Therein, the authors study parallel statistical decisions as in~\eqref{eq:hierarchical}, keeping only the first two levels of the hierarchical model, that is,~\eqref{eq:hierarchy1} and~\eqref{eq:hierarchy2}. The nuisance parameter distribution \smash{$G$} is treated as unknown and estimated as \smash{$\widehat{G} = \widehat{G}(U_1,\dotsc,U_n)$} via the empirical Bayes principle~\citep{robbins1956empirical,efron2019bayes}. The empirical partially Bayes p-values are then computed via the plug-in principle as \smash{$P_i^{\text{EPB}} = \Pvalfunoracle(t, u, \widehat{G})$} with $\Pvalfunoracle$ defined in~\eqref{eq:oracle_G_pvalue}. Here we propose a hierarchical Bayesian counterpart to the empirical Bayes approach.

We finally note that several authors have proposed approaches toward computing Bayesian p-values that are approximately calibrated (uniform) under the null hypothesis, for instance~\citet{hjort2006postprocessing, li2022calibrated, moran2023holdout}. We contribute to this literature by proposing a new construction suitable for the large scale inference problems.~\citet{cademartori2023joint} considers Bayesian p-values when the researcher can construct multiple test statistics for the same hypothesis; by contrast we have multiple hypotheses each with its own test statistic.  In Section~\ref{sec:calibration} below will seek to explain why procedures constructed only imposing a prior on the nuisance parameters may be of interest, and the mechanisms by which they lead to calibrated inference.

\section{Calibration: Sampling frames and asymptotic regimes}
\label{sec:calibration}
\epigraph{``The applied statistician should be Bayesian in principle
and calibrated to the real world in practice---appropriate frequency calculations
help to define such a tie.''}{--- \textup{Donald B. Rubin}, \citeyear{cox1975note}}

In using and intepreting our proposed partially Bayes p-values, an important question is whether they are calibrated, that is, whether $\cpvalue_i \approx \mathrm{Unif}[0,1]$ for $i \in \Hnull$. Answering this question depends on the sampling frame and asymptotic regime in which we study the problem. Here, by sampling frame we mean the following: which of the three levels in the hierarchical model~\eqref{eq:hierarchical} do we treat as fixed in our analysis? 
\begin{itemize}[leftmargin=*]
\item \textbf{Frequentist frame:} We only account for randomness in the first level of the hierarchical specification, i.e., in~\eqref{eq:hierarchy1} and treat all nuisance parameters $\nuisancevector =(\nuisance_1,\dotsc,\nuisance_n)$, as well as the primary parameters $\primaryvector = (\primary_1,\dotsc,\primary_n)$ as fixed parameters. When we derive results in this frame, we use the notation $\nuisancevectortrue$ and $\primaryvectortrue$ for the data-generating parameters and denote probabilities and expectations by $\PP[\nuisancevectortrue, \primaryvectortrue]{\cdot}$ and $\EE[\nuisancevectortrue, \primaryvectortrue]{\cdot}$ (or $\PP[\nuisancetrue_i, \primarytrue_i]{\cdot}$, $\EE[\nuisancetrue_i, \primarytrue_i]{\cdot}$ when the expressions only involve the $i$-th unit).
\item \textbf{Empirical Bayes frame:} Here we also account for randomness in the second level of the hierarchical specification, i.e., \eqref{eq:hierarchy1} and~\eqref{eq:hierarchy2} and treat the nuisance parameters $\nu_i$ as draws from their frequency distribution $G$ which we denote by $\Gstar$. Primary parameters $\primaryvectortrue$ are treated as fixed. We denote probabilities and expectations by $\PP[\Gstar, \primaryvectortrue]{\cdot}$ and $\EE[\Gstar, \primaryvectortrue]{\cdot}$. (The frequentist analysis of Bayesian nonparametric models typically refers to this sampling frame and treats the distribution $\Gstar$ of the latent parameters $\nuisance_i$ as fixed.)
\item \textbf{Hierarchical Bayes frame:} Here we account for randomness in all levels of the hierarchical specification, i.e., in~\eqref{eq:hierarchy1},~\eqref{eq:hierarchy2}, and~\eqref{eq:hierarchy3}. We continue to treat the primary parameters $\primaryvectortrue$ as fixed, denote $\Pi$ by $\Pi^\star$ and write probabilities and expectations as $\PP[\Pi^\star,\primaryvectortrue]{\cdot}$ and $\EE[\Pi^\star,\primaryvectortrue]{\cdot}$.
\end{itemize}
Below, we will study the calibration of the partially Bayes p-values in all of the above frames separately.\footnote{We emphasize that their computation (and definition in~\eqref{eq:conditional_pred_pvalue}) is the same in all cases.}  By iterated expectation, calibration in the frequentist frame implies calibration in the empirical Bayes frame and calibration in the empirical Bayes frame implies calibration in the fully Bayes frame.

Throughout we assume that the test statistics $T_i$ have a continuous distribution.
\begin{assu}[Continuous test statistics]
\label{assu:test_statistics}
The test statistic $T_i$ is such that for all $\nuisance_i \in \nuisancespace$ and $u_i \in \mathcal{U}$, the distribution of $T_i$ conditional on $\nuisance_i$ and $U_i=u_i$ is absolutely continuous with respect to the Lebesgue measure on $\mathbb R$.
\end{assu}
We note that this assumption precludes situations with discrete data. While it is not fundamental to our approach, this assumption will enable us to state results on calibration by comparison to exactly uniformly distributed p-values.
As our first result, we record that
calibration in the fully Bayes frame holds automatically under Assumption~\ref{assu:test_statistics}.
\begin{prop}[Calibration in the fully Bayes frame]
\label{prop:calibration_fully_bayes}
Suppose that data is generated as in the hierarchical model in~\eqref{eq:hierarchical} with $\primaryvectortrue$ fixed and $\Pi= \Pi^{\star}$ and that Assumption~\ref{assu:test_statistics} holds. Then, 
$ \max_{i \in \Hnull} \sup_{\alpha \in [0,1]}\abs{ \PP[\Pi^\star, \primaryvectortrue]{\cpvalue_i \leq \alpha} - \alpha} = 0.$
\end{prop}
\begin{proof}
The proof directly follows by applying the probability integral transform to the distribution of $T_i'$ conditional on $U_1,\dotsc,U_n$. 
\end{proof}
For the result above, it is important that the $\Pi$ in the definition of the partially Bayes p-values in~\eqref{eq:conditional_pred_pvalue} is identical to the data-generating $\Pi^{\star}$.
Proposition~\ref{prop:calibration_fully_bayes} evaluates calibration of the p-values under replications of the data under the full Bayesian hierarchical specification in~\eqref{eq:hierarchical} and assesses the type-I error integrating over the prior predictive distribution. Thus, the result of the theorem is similar in spirit to Theorem 1 and Lemma 1 for the predictive p-values of~\citet{meng1994posterior} with two key differences: incorporating a hierarchical structure with one more level, and getting exact uniformity due to the partial Bayesian conditioning rather than full Bayesian conditioning. 

In so far as the replications assumed in Proposition~\ref{prop:calibration_fully_bayes} are not realistic for practical applications (since we do not expect $\Pi$ to correspond to a data-generating mechanism), we turn to the more interesting question of calibration in the other two sampling frames. For these, we may only hope for calibration to hold asymptotically. Below, we consider two asymptotic regimes under which calibration holds. Our goal is to provide conceptual insights into the asymptotic calibration of the partially Bayes p-values.

\subsection{Representations of partially Bayes p-values}
\label{subsec:representations}
To motivate the asymptotic calibration results, we provide two representations of the partially Bayes p-values in terms of two types of ``oracle'' p-value function. First, consider the oracle p-value function for the $i$-th unit if we knew the precise value of the nuisance parameter $\nuisance_i$~\citep{meng1994posterior},
\begin{equation}
\Pvalfunoracle(t, u, \nuisance) := \PP{\abs{T_i'} \geq \abs{t} \, \mid \, U_i=u,\, \nuisance_i=\nuisance}.
\label{eq:oracle_nuisance_pvalue}
\end{equation}
It is immediate that this oracle p-value function satisfies the following calibration property.
\begin{prop}
  \label{prop:oracle_nuisance_pvalue}
Let $i \in \Hnull$. Suppose that $(T_i, U_i)$ is generated as in the first level of the hierarchical model in~\eqref{eq:hierarchy1} with $\primarytrue_i=\primary_0$ and nuisance parameter $\nuisancetrue_i$ and that Assumption~\ref{assu:test_statistics} holds. Then
$\sup_{\alpha \in [0,1]}\abs{ \PP[\nuisancetrue_i, \primarytrue_i]{\Pvalfunoracle(T_i, U_i, \nuisancetrue_i) \leq \alpha} - \alpha} = 0.$
\end{prop}
Earlier, in~\eqref{eq:oracle_G_pvalue}, we also defined the oracle p-value function $\Pvalfunoracle(t, u, G)$ for fixed $G$.
We use the same notation for the oracle p-value functions in~\eqref{eq:oracle_nuisance_pvalue} and~\eqref{eq:oracle_G_pvalue} and distinguish them according to whether the third function argument is a nuisance parameter $\nuisance \in \mathcal{V}$ or a distribution $G$ over the nuisance parameter. This notation is justified by the fact that $\Pvalfunoracle(t, u, \delta_{\nu})  = \Pvalfunoracle(t, u, \nu)$ where $\delta_{\nu}$ is the Dirac measure at $\nu$. Moreover, the two oracle p-value functions are related to each other via
$
\Pvalfunoracle(t, u_i, G) = \EE[G]{\Pvalfunoracle(t,u_i,\nuisance_i) \mid U_i=u_i}.
$

We next record the calibration property of the oracle p-value function in~\eqref{eq:oracle_G_pvalue}.\footnote{The proof of Propositions~\ref{prop:oracle_nuisance_pvalue} and~\ref{prop:oracle_G_pvalue} follows by applying the probability integral transform to the distribution of $T_i'$ conditional on $U_i$ and $\nuisancetrue_i$, respectively $U_i$ and $\Gstar$. We omit details.
}
\begin{prop}
  \label{prop:oracle_G_pvalue}
Let $i \in \Hnull$. Suppose that $(T_i, U_i)$ is generated as in the first two levels of the hierarchical model in~\eqref{eq:hierarchical} with $\primarytrue_i=\primary_0$, $\nuisance_i \sim \Gstar$ and that Assumption~\ref{assu:test_statistics} holds. Then,
$\sup_{\alpha \in [0,1]}\abs{ \PP[\Gstar, \primarytrue_i]{\Pvalfunoracle(T_i, U_i, \Gstar) \leq \alpha} - \alpha} = 0.$
\end{prop}

The posterior mean of the distribution $G$ given $U_1,\dotsc,U_n$ under the prior $\Pi$ is the probability measure on $\nuisancespace$ that assigns to each measurable set $A \subseteq \nuisancespace$ the probability
\begin{equation}G_{\Pi}[u_1,\ldots,u_n](A) := \EE[\Pi]{G(A) \mid U_1 = u_1,\ldots,U_n=u_n}. 
    \label{eq:posterior_G}
\end{equation}
The main result of this section is that the partially Bayes p-values may be represented in terms of the oracle p-value functions in~\eqref{eq:oracle_G_pvalue} and~\eqref{eq:oracle_nuisance_pvalue}. These representations will be crucial for understanding the asymptotic behavior of the proposed partially Bayes p-values.

\begin{theo}[Partially Bayes and oracle p-value functions]
\label{theo:oracle_representation}
The $i$-th partially Bayes p-value function $\cpvaluefun_i(t, (u_1,\dotsc,u_n);\, \Pi)$ is equal to all of the following three expressions:
\begin{enumerate}[label=(\alph*)]
  \item $\EE[\Pi]{\Pvalfunoracle(t, u_i, \nu_i) \,\mid \, U_1 = u_1,\dotsc, U_n=u_n}$,
  \item  $\EE[\Pi]{\Pvalfunoracle(t, u_i, G) \,\mid \, U_1 = u_1,\dotsc, U_n=u_n}$,
  \item[(b')] $\Pvalfunoracle(t, u_i, G_{\Pi}[u_1,\ldots,u_{i-1},u_{i+1}, \ldots,u_n])$.
\end{enumerate}
\end{theo}

\begin{proof}
The results in (a) and (b) follow by iterated expectation, conditioning on $U_1,\dotsc,U_n, \nu_i$ and $U_1,\dotsc,U_n, G$ respectively. The result in (b') follows by first updating $G$ conditional on $U_j, j \neq i$ (leaving $U_i$ out) and then evaluating the oracle p-value function $\Pvalfunoracle(t, u_i, G)$ at the posterior mean distribution in~\eqref{eq:posterior_G}. 
\end{proof}
For (a), the partially Bayes p-value $\cpvalue_i$ may be interpreted as the expectation of the oracle p-value function $\Pvalfunoracle(t, u_i, \nu_i)$ with respect to the posterior distribution of the $i$-th nuisance parameter $\nu_i$. 
The results in (b) and (b') are conceptually similar and relate the partially Bayes p-values to the empirical partially Bayes p-values studied by~\citet{ignatiadis2024empirical}. Therein, one first estimates $G$ by \smash{$\widehat{G}$}, and then 
computes \smash{$\ebpvaluefun_i = \Pvalfunoracle(T, U_i, \widehat{G})$}. In this way, following the Bayes empirical Bayes maxim of~\citet{deely1981bayes}, one could have used the terminology ``Bayes empirical partially Bayes p-values'' instead of partially Bayes p-values. We prefer the latter as it is shorter and our proposal may be interpreted without reference to the empirical Bayes principle.

\subsection{Asymptotic regime I: Information-rich statistical tasks}

As previewed earlier, we consider two asymptotic regimes for our calibration results.
In the first regime, we keep the tasks $i=1,\dotsc,n$ fixed and consider asymptotics in which the information for the tasks increases. To be concrete, we write \smash{$\dataset_i = \dataset_i^{K}$} where $K \in \mathbb N$ is a parameter that controls the amount of information in the data and we will let $K \to \infty$. If $\dataset_i$ consists of iid observations then $K$ would be the sample size, but we allow for more general interpretation of $K$. We use $K$ as a superscript for the primary and nuisance parameter statistics $T_i = T_i^{K}$,  $U_i = U_i^{K}$, and also allow the space in which $U_i$ takes values to depend on $K$, that is, $U_i \in \mathcal{U}^K$. Finally, we also use $K$ as a subscript for the oracle p-value functions $\Pvalfunoracle(t,u,\nu) = \Pvalfunoracle[K](t,u,\nu)$. The parameters $\primary_i$ and $\nuisance_i$ do not depend on $K$.

In this regime, we may expect the partially Bayes p-values to be calibrated as $K \to \infty$ in the frequentist frame by the following reasoning. Suppose the posterior for $\nuisance_i$ given \smash{$U_1^{K},\dotsc,U_n^{K}$} concentrates around $\nuisancetrue_i$. For instance, a Bernstein-von Mises theorem may indicate that the posterior distribution would concentrate around \smash{$\hat{\nuisance}_i = \hat{\nuisance}(U_i)$}, the maximum likelihood estimator of $\nuisance_i$ based on \smash{$U_i = U_i^K$}. Then, we may expect the following approximate equalities,
$$
\EE[\Pi]{\Pvalfunoracle(t, U_i, \nu_i) \,\mid \, U_1 ,\dotsc, U_n} \approx \Pvalfunoracle(t, U_i, \hat{\nuisance}_i) \approx \Pvalfunoracle(t, U_i,\nuisancetrue_i),
$$
where the first approximate equality follows from the assumed concentration of the posterior distribution of $\nuisance_i$ around $\hat{\nuisance}_i$, and the second assumes that $\hat{\nuisance}_i$ is a good estimate of $\nuisancetrue_i$ (because $K$ is large). Thus, using Theorem~\ref{theo:oracle_representation}a), we get that $\cpvalue_i \approx \Pvalfunoracle(T_i, U_i,\nuisancetrue_i)$, and the latter is calibrated by Proposition~\ref{prop:oracle_nuisance_pvalue}.

To simplify statements of results, we make the following regularity assumptions.

\begin{assu}[Nuisance parameters]
\label{assu:nuisance_parameters}
The nuisance parameter space $\nuisancespace$ is a compact metric space with metric $d(\cdot,\cdot)$. 
\end{assu}

\begin{assu}[Regular p-value functions]
  \label{assu:regular_pvalues}
The functions $\cb{\mathcal{V} \to [0,1],\, \nu \mapsto \Pvalfunoracle[K](t^K,u^K,\nu)}$ indexed by $K \in \mathbb N$, $t^{K} \in \RR, u^{K} \in \mathcal{U}^{K}$ are uniformly equicontinuous.
\end{assu}

The following theorem records a formal statement regarding asymptotic calibration in the frequentist frame as $K \to \infty$.

\begin{theo}[Calibration as $K\to \infty$ in the frequentist frame]
\label{theo:K_to_infty}
Let $\nuisancetrue_i \in \mathcal{V}$, $i=1,\dotsc,n$ be fixed. Suppose that Assumptions~\ref{assu:test_statistics},~\ref{assu:nuisance_parameters} and~\ref{assu:regular_pvalues} hold, and that for all $i \in \Hnull$ we have that:
\begin{itemize}
\item[$(*)$] For any $\delta >0$, it holds that 
$\EE[\nuisancevectortrue]{ \Pi(\nu_i\,:\, d(\nuisance_i, \nuisance_i^{\star}) > \delta \mid U_1,\ldots,U_n)} \to 0\, \text{ as }\, K \to \infty.$    
\end{itemize}
Then, as $K \to \infty$ ($n$ fixed), we have that:\;\;
$\displaystyle\limsup_{K \to \infty} \max_{i \in \Hnull} \sup_{\alpha \in [0,1]}\abs{\PP[\nuisancevectortrue, \primaryvectortrue]{ \cpvalue_i \leq \alpha} - \alpha} =0.$
\end{theo}
Assumption $(*)$ posits the key requirement for the result: for each $i \in \Hnull$, the posterior for $\nuisance_i$ concentrates around $\nuisancetrue_i$ (i.e., we have posterior consistency for $\nuisancetrue_i$). 
Under the above assumptions, our proposed partially Bayes p-values are calibrated asymptotically in $K$. Note that the above statements are purely frequentist. The result above is not surprising. For instance, the conditional predictive p-values of~\citet{bayarri2000values} (that do not share information across problems) are asymptotically calibrated as shown in~\citet{robins2000asymptotic}.

In Proposition~\ref{prop:K_to_infty_empirical_bayes} of Supplement~\ref{sec:further_theoretical_results} we show that 
the calibration result of Theorem~\ref{theo:K_to_infty} extends to the empirical Bayes frame.

\subsection{Asymptotic regime II: Large number of statistical tasks}

Our second asymptotic regime is less standard and clarifies the benefits of large scale inference.
Here, we keep the amount of information for each statistical task fixed, say at $K_0$ (which we omit from the notation) and consider asymptotics with a growing  number of statistical tasks, i.e., $n \to \infty$. In this regime, there are qualitative differences between the frequentist and empirical Bayes frames, and we start by discussing calibration in the empirical Bayes frame.

\subsubsection{Empirical Bayes frame}
In the asymptotic regime $n \to \infty$, we may expect the partially Bayes p-values to be calibrated in the empirical Bayes frame by the following reasoning. Suppose that the nuisance parameter distribution $G$ is estimated consistently based on $U_1,\dotsc,U_n$ in the sense that the posterior is weakly consistent at $\Gstar$ as $n \to \infty$. Then we would get that,
\begin{equation}
\EE[\Pi]{\Pvalfunoracle(t, U_i, G) \,\mid \, U_1,\dotsc, U_n} \approx \Pvalfunoracle(t, U_i, \Gstar),
\label{eq:approximation_n_to_infty}
\end{equation}
and so by Theorem~\ref{theo:oracle_representation}b), we would have that $\cpvalue_i \approx \Pvalfunoracle(t, U_i, \Gstar)$. The latter is calibrated by Proposition~\ref{prop:oracle_G_pvalue}, and so $\cpvalue_i$ should also be asymptotically calibrated in the empirical Bayes frame. For the formal statement, we introduce one additional technical assumption.
\begin{assu}[Nuisance parameter statistics]
  \label{assu:nuisance_statistics}
The collection of distributions given by $\cb{ \mathbb P[U_i \in \cdot \mid \nuisance_i] : \nuisance_i \in \nuisancespace}$ is tight, that is, for all $\varepsilon > 0$, there exists a compact set $C \subset \mathcal{U}$ such that $\PP[\nuisance_i]{U_i \notin C} < \varepsilon$ for all $\nuisance_i \in \nuisancespace$. All distributions in the same collection are absolutely continuous with respect to a measure $\lambda$ on $\mathcal{U}$ and their density is denoted by $p(u \mid \nu)$. For any compact set $C \subset \mathcal{U}$, we have that $\lambda(C) < \infty$ and the functions $\cb{\nu \mapsto p(u \mid \nu) : u \in C}$ are uniformly bounded and equicontinuous.
\end{assu}

To state our theorem, we also recall the definition of the bounded Lipschitz metric $\Dbl$ on the space of all probability measures $\mathcal{P}(\mathcal{V})$ supported on the metric space $\mathcal{V}$. For $G, G' \in \mathcal{P}(\mathcal{V})$, we define
\begin{equation*}
  \label{eq:BL_metric}
\Dbl(G, G') := \sup \cb{ \abs{\int \psi(\nu) \, G(\dd\nu) - \int \psi(\nu) \, G'(\dd\nu)}\;:\; \Norm{\psi(\cdot)}_{\infty} \leq 1,\; \psi(\cdot) \mbox{~is 1-Lipschitz} }.
\end{equation*}
Our main result in this regime is the following theorem.

\begin{theo}[Calibration as $n \to \infty$ in the empirical Bayes frame]
\label{theo:calibration_n_to_infty_empirical_bayes}
  Suppose that $K$ remains fixed and that $n \to \infty$. Let $\Gstar$ in~\eqref{eq:hierarchy2} be fixed, \smash{$\nuisance_i \simiid \Gstar$}, suppose Assumptions~\ref{assu:test_statistics}, ~\ref{assu:nuisance_parameters},~\ref{assu:regular_pvalues} and~\ref{assu:nuisance_statistics} hold and further assume:
  \begin{itemize}
  \item [$(*)$] $\EE[\Gstar]{\mathfrak{D}_{\mathrm{BL}}(G_{\Pi}[U_1,\ldots,U_n],\, \Gstar)} \to 0 \,\text{ as }\, n \to \infty.$
  \end{itemize}
  Then as $n \to \infty$ (with $K$ fixed), it holds that:
  $$\limsup_{n \to \infty} \sup_{\alpha \in [0,1]}  \max_{i \in \Hnull} \abs{\PP[\Gstar, \primaryvectortrue]{ \cpvalue_i \leq \alpha} -\alpha} =0.$$
\end{theo}

Condition $(*)$ of the theorem requires that the posterior mean of $G$ defined in~\eqref{eq:posterior_G} is consistent for $\Gstar$. By~\citet[Theorem 6.8]{ghosal2017fundamentals}, this condition is implied by consistency of the posterior distribution of $G$ at $\Gstar$, that is, if for all $\delta>0$, $\EE[\Gstar]{\Pi(G\,:\, \Dbl(G, G_{\star})>\delta \mid U_1,\ldots,U_n)} \to 0$ as $n \to \infty$, then $(*)$ holds. 
We note that both $(*)$ and its sufficient condition are deconvolution results for $\Gstar$, since we only observe indirect measurements of $\nuisance_i$ through $U_i$. Such deconvolution results are derived in special cases, e.g., by~\citet{nguyen2013convergence, scricciolo2018bayes, su2020nonparametric, rousseau2021wasserstein}.
In verifying whether consistency as above holds for $\Gstar$, it may be easier to use results on posterior consistency for the marginal distribution of $U_i$ (of which we have direct measurements). We provide a simple, generic result of this type below that requires consistency for the marginal distribution in the stronger total variation metric as in e.g.,~\citet{ghosal2001entropies}.

For any distribution $G$ supported on $\nuisancespace$, define the $\dd   \lambda$-marginal density of $U_i$ as
\begin{equation}
f(u; G) := \int_{\nuisancespace} p(u \mid \nu) \, G(\dd\nu).
\label{eq:marginal_density}
\end{equation}
The total variation distance of two distributions on $\mathcal{U}$ with $\dd\lambda$-densities $f_1$, $f_2$ is defined as
$$
\TV(f_1,f_2) := \frac{1}{2} \int_{\mathcal{U}} \abs{f_1(u) - f_2(u)} \, \lambda(\dd u).
$$
\begin{prop} 
\label{prop:TV_to_BL}
Suppose that for any $\delta>0$ it holds that $$\EE[\Gstar]{\Pi(G\,:\, \TV(f(\cdot; G)  , f(\cdot; G_{\star}))>\delta \mid U_1,\ldots,U_n)} \to 0\, \text{ as }\, n \to \infty.$$
Suppose moreover that identifiability holds, i.e., that $f(\cdot;G) = f(\cdot,G')$ for $G,G'$ distributions supported on $\nuisancespace$ implies that $G=G'$.
Then for any $\delta>0$ it also holds that
$$\EE[\Gstar]{\Pi(G\,:\, \Dbl(G  , G_{\star})>\delta \mid U_1,\ldots,U_n)} \to 0\, \text{ as }\, n \to \infty,$$
and thus condition $(*)$ of Theorem~\ref{theo:K_to_infty} holds.
\end{prop}

To recap, Theorem~\ref{theo:calibration_n_to_infty_empirical_bayes} shows that the partially Bayes p-values provide calibrated inference  under mild assumptions. Herein it is important to emphasize that the notion of consistency we require for $\Gstar$ is only with respect to weak convergence. By contrast, analogous results for e.g., credible intervals, would require stronger notions of consistency, e.g., in total variation distance.\footnote{This statement is not in contradiction with the sufficient condition in Proposition~\ref{prop:TV_to_BL}. Therein, the assumption is on the marginal distribution of $U_i$ and not on $G$.} As is known from the classical results of~\citet{kiefer1956consistency} (in the context of nonparametric maximum likelihood), weak consistency in deconvolution problems is in general possible under mild assumptions---this is not true for the total variation distance.

Our partially Bayes p-values provide a pragmatic construction that reaps benefits of Bayesian nonparametrics without requiring joint modeling of both nuisance and primary parameters which may require substantially more involved and careful modeling as in, e.g,
\citet{dahl2009spiked, denti2021two}. Moreover, coverage guarantees for Bayes credible intervals in large scale inference are difficult to come by, even in problems without nuisance parameters~\citep{vanderpas2017uncertainty}, and may require inflating the width of the intervals.

\subsubsection{Frequentist frame: compound p-values}

Our final calibration result pertains to a purely frequentist analysis of the regime in which $n\to \infty$ and the information per task remains fixed. In this case there is no data-generating distribution $\Gstar$, so we can no longer argue that $G_{\Pi}[U_1,\ldots,U_n] \approx \Gstar$ (since $\Gstar$ is not defined). Nevertheless, according to results in compound decision theory~\citep{robbins1951asymptotically, zhang2003compound}, we may expect that $G_{\Pi}[U_1,\ldots,U_n]  \approx G(\nuisancevectortrue)$, where
\begin{equation}
G(\nuisancevectortrue) := \frac{1}{n} \sum_{i=1}^n \delta_{\nuisancetrue_i}
\label{eq:empirical_distribution_of_nuisance}
\end{equation}
is the empirical distribution of the nuisance parameters $\nuisancetrue_1,\dotsc,\nuisancetrue_n$. In this case, in analogy to~\eqref{eq:approximation_n_to_infty}, it would follow that:
\begin{equation}
\EE[\Pi]{\Pvalfunoracle(t, U_i, G) \,\mid \, U_1,\dotsc, U_n} \approx \Pvalfunoracle(t, U_i, G(\nuisancevectortrue)).
\label{eq:approximation_n_to_infty_compound}
\end{equation}
Thus, to understand the asymptotic calibration of $\cpvalue_i$ in the frequentist frame, we first examine the properties of $\Pvalfunoracle(T_i, U_i, G(\nuisancevectortrue))$. While these oracle quantities are not standard p-values (i.e., not uniformly distributed under the null for a single hypothesis), the following result—--a frequentist counterpart to Proposition~\ref{prop:oracle_G_pvalue}---shows they are compound p-values as defined in~\citet{ignatiadis2025asymptotic}.

\begin{prop}[Compound p-values]
  \label{prop:compound_pvalues}
Suppose that $(T_i, U_i)$, $i=1,\ldots,n$, are generated as in the first level of the hierarchical model in~\eqref{eq:hierarchy1} with primary parameters $\primaryvectortrue$ and nuisance parameters $\nuisancevectortrue$. Then, $\Pvalfunoracle(T_i, U_i, G(\nuisancevectortrue)),\;i=1,\ldots,n,$ are compound p-values, that is,
$$
\sup_{\alpha \in [0,1]} \p{ \frac{1}{n}\sum_{i \in \Hnull} \PP[\nuisancevectortrue,\primaryvectortrue]{\Pvalfunoracle(T_i, U_i, G(\nuisancevectortrue)) \leq \alpha} - \alpha}_+ \leq 0,
$$
where $(x)_+ := \max(x,0)$.
\end{prop}
The proof is analogous to the proof of Theorem 21 in~\citet{ignatiadis2024empirical} and is omitted.
The compound p-value property in Proposition~\ref{prop:compound_pvalues} provides a useful frequentist guarantee, connecting to the analysis in Section~\ref{subsec:industrialist}. In the original study, \citet{palmieri2015genomewide}  report as significant all genes with a t-test p-value less than or equal to $0.001$.~\citet{klaus2018end} observe that, while the fixed-threshold procedure of~\citeauthor{palmieri2015genomewide} is not a genuine multiple testing correction, it provides a bound on the expected number of false discoveries ($\alpha \cdot n \approx 16$). Proposition~\ref{prop:compound_pvalues} shows that the oracle quantities $\Pvalfunoracle(T_i, U_i, G(\nuisancevectortrue))$ (although not genuine p-values) provide the same guarantee on the expected number of false discoveries. This is because they satisfy the compound p-value property--- also introduced earlier under the name average significance control by~\citet{armstrong2022false}---meaning their average null distribution (normalized by $n$) is stochastically larger than uniform, without requiring uniformity for individual p-values.

We also briefly remark that when compound p-values are used in conjunction with the procedure of~\citet{benjamini1995controlling}, then false discovery rate (FDR) is asymptotically controlled under certain regimes~\citep{armstrong2022false, ignatiadis2025asymptotic} and if the compound p-values are jointly independent, then it is also controlled in finite samples~\citep{barber2025false} with a mild inflation of the target FDR level.

The following theorem shows that indeed the partially Bayes p-values are asymptotically compound p-values in the frequentist frame.

\begin{theo}[Calibration as $n \to \infty$ in the frequentist frame]
\label{theo:calibration_n_to_infty_frequentist}
  Suppose that $K$ remains fixed and that $n \to \infty$. Suppose Assumptions~\ref{assu:test_statistics}, ~\ref{assu:nuisance_parameters},~\ref{assu:regular_pvalues} and~\ref{assu:nuisance_statistics} hold. Write $\nuisancevectortrue_{n,-i} = (\nuisancetrue_1,\ldots,\nuisancetrue_{i-1},\nuisancetrue_{i+1},\ldots,\nuisancetrue_n)$ for the nuisance parameters excluding the $i$-th one and assume
    \begin{itemize}
  \item [$(*)$] $\displaystyle\max_{i=1,\ldots,n }\EE[\nuisancevectortrue_{n,-i}]{\mathfrak{D}_{\mathrm{BL}}(G_{\Pi}[U_1,\ldots,U_{i-1},U_{i+1},\ldots U_n],\, G(\nuisancevectortrue_{n,-i}))} \to 0 \,\text{ as }\, n \to \infty.$
  \end{itemize}
  Then as $n \to \infty$ (with $K$ fixed), $\cpvalue_1,\ldots,\cpvalue_n$ are asymptotic compound p-values, i.e., 
  $$\limsup_{n \to \infty} \sup_{\alpha \in [0,1]}  \p{\frac{1}{n} \sum_{i \in \Hnull} \PP[\nuisancevectortrue, \primaryvectortrue]{ \cpvalue_i \leq \alpha} -\alpha}_+ \leq 0.$$
\end{theo}
Assumption $(*)$ of the theorem is analogous to the corresponding assumption of Theorem~\ref{theo:calibration_n_to_infty_empirical_bayes}, except stated in a leave-one-out fashion and targeting the empirical distribution of nuisance parameters in~\eqref{eq:empirical_distribution_of_nuisance} rather than $\Gstar$ (which is not defined in the frequentist frame). High-level conditions under which $(*)$ holds are provided in Theorem 1 of~\citet{datta1991consistency}.

\section{Normal means with unknown and varying variance}
\label{sec:normal_means}

We now turn to concrete applications of our proposed partially Bayes p-values. Our first application pertains to the following normal means problem. For the $i$-th unit, we observe 
\begin{equation}
\label{eq:normal_means}
\dataset_i = \cb{Z_{i1}, \dotsc, Z_{iK}},\;\;\; Z_{ij} \simiid \mathrm{N}(\mu_i, \sigma_i^2).
\end{equation}
Here, the primary parameter is the mean, $\primary_i = \mu_i$, and we want to test $H_i: \mu_i=0$. The nuisance parameter is the variance, $\nuisance_i = \sigma_i^2$. By sufficiency, we may collapse $\dataset_i$ into the sample average and sample variance. That is, letting $\bar{Z}_i := \sum_{j=1}^{K}Z_{ij}/K$. we may take 
\begin{equation}
\label{eq:normal_means_summary_stats}
T_i := \sqrt{K}\bar{Z}_i \sim \mathrm{N}(\sqrt{K}\mu_i, \sigma_i^2),\;\;\; U_i \equiv S_i^2:= \frac{1}{K-1} \sum_{j=1}^{K}\p{Z_{ij} - \bar{Z}_i}^2 \sim \sigma_i^2 \frac{ \chi^2_{K-1}}{K-1}.
\end{equation}
\begin{rema}[Standard t-test]
\label{rema:standard_t_test}
Suppose that instead of the summarization in~\eqref{eq:normal_means_summary_stats} we had summarized $\dataset_i$ as $T_i$ and \smash{$\widetilde{U}_i := \sum_{j=1}^{K} Z_{ij}^2/K$}. It follows by~\citet{bayarri2000values} that
\smash{$\cpvalue_i = 2F_{t,K-1}(-\abs{T_i}/\widetilde{U}_i^{1/2}),$}
where $F_{t,K-1}$ is the distribution function of the t-distribution with $K-1$ degrees of freedom.
In words, the partially Bayes p-values are identical to the p-values from the standard t-test. Thus $\cpvalue_i$ does not depend on the prior distribution $\Pi$ in~\eqref{eq:hierarchy3} and no sharing of information occurs across units $i$. Instead, we propose to apply our framework using $U_i=S_i^2$ in~\eqref{eq:normal_means_summary_stats}, which is ancillary for the primary parameter $\mu_i$.
\end{rema}

The oracle p-value function~\eqref{eq:oracle_nuisance_pvalue} in this setting is $\Pvalfunoracle(t, u, \nuisance) = 2\Phi(-|t|/\sigma)$, where $\nuisance=\sigma^2$ and $\Phi$ is standard normal distribution function. The oracle p-value function in~\eqref{eq:oracle_G_pvalue} will in general depend on the distribution $G$. In the special case $G=\mathrm{inv}\chi^2(\nu_0, \sigma_0^2)$, we have that,
$$
\Pvalfunoracle(t, u, \mathrm{inv}\chi^2(\nu_0, \sigma_0^2)) = 2F_{t, K-1 + \nu_0}\p{ -|t| \Big/ \sqrt{\frac{(K-1)u + \nu_0 \sigma_0^2}{ (K-1)+\nu_0}}},
$$
where $F_{t, K-1 + \nu_0}$ is the cumulative distribution function of the t-distribution with $K-1+\nu_0$ degrees of freedom. The interpretation of the above is that by imposing an informative prior on the nuisance parameters (the variances), we are able to increase the degrees of freedom of the t-test (from $K-1$ to  $K-1+\nu_0$), shrink estimates of the variances towards a common value and increase power compared to the ``standard'' t-test of Remark~\ref{rema:standard_t_test}.
The p-values of the well-known and widely used limma moderated t-test of~\citet{smyth2004linear} are precisely given by $\Pvalfunoracle(t, u, \mathrm{inv}\chi^2(\hat{\nu}_0, \hat{\sigma}_0^2))$ where $\hat{\nu}_0, \hat{\sigma}_0^2$ are empirical Bayes estimates of the hyperparameters $\nu_0, \sigma_0^2$. \citet{ignatiadis2024empirical} study the empirical partially Bayes p-values \smash{$\Pvalfunoracle(T_i, U_i, \widehat{G})$} where \smash{$\widehat{G}$} is a nonparametric maximum likelihood estimate of $G$.

In this work, instead we propose to compute partially Bayes p-values for this problem by placing a nonparametric prior $\Pi$ on the nuisance parameter distribution $G$, 
\begin{equation}
  \label{eq:DP}
\sigma_i^2 \simiid G,\;\;\; G \sim \mathrm{DP}(c, G_0),
\end{equation}
where $\mathrm{DP}(c, G_0)$ is the Dirichlet Process prior~\citep{ferguson1973bayesian} with concentration parameter $c>0$ and a base measure $G_0$. This modeling choice allows for both clustering and heterogeneity of the variances across the different units.

In our implementation, we set $c$ to have hyperprior, $c \sim \text{Gamma}(0.001, 100)$ (shape-scale parameterization) and we set the base distribution $G_0 = \mathrm{inv}\chi^2(\hat{\nu}_0, \hat{\sigma}_0^2)$ where we choose $(\hat{\nu}_0, \hat{\sigma}_0^2)$ such that the 1st and 99th percentiles of $G_0$ match \smash{$\min_i\{ S_i^2\}$} and \smash{$\max_i\{S_i^2\}$}, respectively. This choice yields a diffuse base measure and allows for conjugate updates (see below).

\begin{rema}[Parametric partially Bayes Limma]
As mentioned above, limma implements an empirical partially Bayes procedure with the estimated prior $\mathrm{inv}\chi^2(\hat{\nu}_0, \hat{\sigma}_0^2)$. The procedure we described in this section modifies limma in two ways: using a nonparametric specification for the prior on the variances and replacing the empirical Bayes estimate with hierarchical Bayes. An intermediate approach is to keep the parametric specification of limma, but to replace the empirical Bayes estimate with a hierarchical Bayes approach.\footnote{An alternative intermediate approach is to use a nonparametric specification for the prior along with empirical Bayes. This is the approach pursued by~\citet{ignatiadis2024empirical}.
} We could form such parametric partially Bayes p-values by placing a hyperprior on $\nu_0, \sigma_0^2$ in the inverse-chi-squared prior.
\end{rema}

\begin{rema}[Two-sample problem]
Consider the two sample problem with unequal variances. For the $i$-th unit we observe $\dataset_i = \{Z_{i1}, \dotsc, Z_{iK}, Y_{i1}, \dotsc, Y_{iL}\}$ where
\smash{$Z_{ij} \simiid \mathrm{N}(\mu_{i1}, \sigma_{i1}^2)$} and \smash{$Y_{ij} \simiid \mathrm{N}(\mu_{i2}, \sigma_{i2}^2)$}. We seek to test $H_i: \primary_i=0$ where $\primary_i = \mu_{i1}-\mu_{i2}$ and the nuisance parameters are the variances, i.e., $\nuisance_i=(\sigma_{iA}^2,\sigma_{iB}^2)$. Herein we can use $T_i=\bar{Z}_i - \bar{Y}_i$ as the test statistic and $U_i=(S_{iZ}^2, S_{iY}^2)$ as the nuisance statistic where $S_{iZ}^2$ and $S_{iY}^2$ are the sample variances of the $Z_{ij}$'s and $Y_{ij}$'s, respectively. In this context,~\citet{ling2025empirical} propose an empirical partially Bayes procedure. Instead, by imposing a nonparametric prior on the joint distribution of $(\sigma_{i1}^2, \sigma_{i2}^2)$, we can compute partially Bayes p-values.
\end{rema}

\subsection{Computation}
\label{subsec:computation_normal_means}
Posterior inference in the model given by the Dirichlet Process prior in~\eqref{eq:normal_means_summary_stats}
and the distribution of $U_i$ in~\eqref{eq:DP} is standard, and so we only provide a brief summary here in so far as it pertains to the computation of the partially Bayes p-values. 

Our implementation follows Neal's Algorithm 2~\citep{neal2000markov} for sampling from Dirichlet process mixture models along with the concentration parameter updates of~\citet{escobar1995bayesian}. This is a Gibbs sampler that maintains the following state variables at MCMC iteration $b$:
\begin{itemize}[noitemsep,leftmargin=*]
\item cluster assignments $c_1^{(b)}, \dotsc, c_n^{(b)}$ where $c_i^{(b)} \in \{1,\dotsc,K^{(b)}\}$ and $K^{(b)}$ is the number of occupied clusters at iteration $b$;
\item cluster variance parameters $\sigma_1^{2(b)}, \dotsc, \sigma_{K^{(b)}}^{2(b)}$;
\item concentration parameter $c^{(b)}$ of the Dirichlet process.
\end{itemize}
The updates above are simplified by the conjugacy of the base distribution $G_0$ to the likelihood of $U_i$. We defer more details to Supplement~\ref{subsec:mcmc_normal_means}. After $B$ full iterations (post burn-in), we approximate the partially Bayes p-values as:
\begin{equation}
\label{eq:pval_normal_means}
\cpvalue_i \approx \frac{1}{B} \sum_{b=1}^B \Pvalfunoracle\p{T_i, U_i, \sigma_{c_i^{(b)}}^{2(b)}} = 
\frac{2}{B} \sum_{b=1}^B \Phi\left(-\abs{T_i}\,\Big /\,\sigma_{c_i^{(b)}}^{(b)}\right),
\end{equation}

\section{Location problems with unknown noise distribution}
\label{sec:unknown_shape}

The normality assumption in~\eqref{eq:normal_means} is strong. The rationale for positing it, is that in common application in genomics, $K$ may be very small, and so, a nonparametric test of the location being equal to $0$ would be nearly powerless. (In fact, as explained in Section~\ref{sec:normal_means}, even the standard t-test may not be powerful enough for common settings in genomics.)

In this section we ask the following question: can we dispense with the parametric normality assumption in~\eqref{eq:normal_means}, while still deriving a powerful procedure?
Our starting point is the following model for the $i$-th observed dataset
\begin{equation}
\label{eq:location_unknown_shape}
\dataset_i = \cb{Z_{i1}, \dotsc, Z_{iK}},\;\;\; Z_{ij}  \simiid  W_i(\cdot - \mu_i),
\end{equation}
where $W_i$ is an unknown noise distribution with location centered at $0$. Here, for simplicity, we assume that the centering is defined by requiring $\int t  W_i(\dd t) =0$, so that the location parameter is the mean of $Z_{ij}$ (i.e., $\EE[W_i,\mu_i]{Z_{ij}} = \mu_i$), however, other definitions of centering (e.g., median) can be handled similarly.

The primary parameter is the location, $\primary_i = \mu_i$, and the nuisance parameter is the shape of the noise distribution, $\nuisance_i = W_i(\cdot)$. We propose to summarize the data as
\begin{equation}
\label{eq:summary_general_loc}
T_i = \bar{Z}_i, \;\;\;\; U_i = (Z_{i1}- \bar{Z}_i, \dotsc, Z_{iK} - \bar{Z}_i).
\end{equation}
The nuisance parameter statistic $U_i$ is ancillary for the location parameter $\theta_i$ and it is known as the configuration statistic of~\citet{fisher1934twoa}. According to Fisher, inference in location models (with known noise distribution, that is, with $W_i$ in~\eqref{eq:location_unknown_shape} known), needs to proceed \emph{conditional} on $U_i$. It is also well-known~\citep{pitman1939estimation, fraser2004ancillaries} how to conduct such conditional inference. Let $w_i$ denote the Lebesgue density of $W_i$, then
\begin{equation}
  \Pvalfunoracle(T_i, U_i, W_i) = \int_{\abs{t} \geq \abs{T_i}} p(t \mid U_i, \primary_i=0) \dd t,\; \text{ where }\, p(t \mid U_i, \primary_i=0) \propto \prod_{j=1}^K w_i(t + U_{ij}).
\label{eq:config_pvalue}
\end{equation}
Above, $p(t \mid U_i, \primary_i=0)$ is the conditional density of $T_i$ given $U_i$ under the null ($\primary_i=0$). Note that $ \Pvalfunoracle(T_i, U_i, W_i)$ in~\eqref{eq:config_pvalue} can be computed by applying 1-dimensional quadrature twice; once to compute the normalizing constant for $p(t \mid U_i, \primary_i=0)$, and then to compute the tail area.
Yet, such approaches are not used in practice because the requirement that $W_i$ is known is too strong (indeed, even with normal data as in~\eqref{eq:normal_means}, one would need to also know the variance $\sigma_i^2$). \citet{severini1994nonparametric} and~\citet{marden2000comment} propose heuristic approaches to compute such conditional p-values in the absence of knowledge of $W_i$. Here we argue that Fisher's envisioned conditioning can be carried out using partially Bayes p-values.

To carry out our principle, we need to put a Bayesian nonparametric prior on the noise distribution $W_i$. We propose:
\begin{equation}
\label{eq:polya_dp}
\begin{aligned}
&W_i(\cdot) = W(\cdot/\tau_i),\;\;\sigma_i^2=\tau_i^2 \textstyle\int u^2 W(\dd u),\;\; \\
& W \sim \mathrm{SymmPT}(\mathcal{A}, G_0^W, J),\;\;\sigma_i^2 \simiid  G^{\Sigma},\; G^{\Sigma} \sim \mathrm{DP}(c, G_0^{\Sigma}),
\end{aligned}
\end{equation}
with $W$ and $(\sigma_i^2)$ independent. In words, we assume that the shape of all $W_i$ is identical and equal to $W$, where $W$ is a draw from a symmetrized Pólya tree (PT)~\citep{lavine1992aspects}, a flexible model for symmetric noise distributions that we describe further below.\footnote{If we change our primary parameter to the median rather than the mean, then we can dispense with the symmetry assumption on $W$ by using a standard (not symmetrized) PT with median $0$~\citep{walker1999bayesian}.
} We introduce further heterogeneity across units by assuming that the variance of each is drawn from a Dirichlet Process, similar to~\eqref{eq:DP} of Section~\ref{sec:normal_means}.

The $\mathrm{SymmPT}(\mathcal{A}, G_0^W, J)$ prior is defined as follows (also see Supplement~\ref{sec:polya_trees_details}). First, we draw  \smash{$\widetilde{W} \sim \mathrm{PT}(\mathcal{A}_0, G_0^W, J)$}, Then we set $W$ to be the symmetrized version of \smash{$\widetilde{W}$}, that is, \smash{$W(A) = \{\widetilde{W}(A) + \widetilde{W}(-A)\}/2$} for all measurable sets $A$. The symmetrization ensures that $W$ is symmetric around $0$.
To sample \smash{$\widetilde{W}$} from the {P\'olya} Tree prior $\mathrm{PT}(\mathcal{A}_0, G_0^W, J)$ (truncated up to level $J$) with base distribution $G_0^W$ and Beta parameters $\mathcal{A}_0$ (consisting of numbers $\alpha(j,\ell)>0$ with $\ell=1,\dotsc,2^j$, $j=1,\dotsc,J$) one proceeds as follows. For all $j$ and odd $\ell$ one draws independently $\beta(j,\ell) \sim  \mathrm{Beta}(\alpha(j,\ell), \alpha(j,\ell+1))$ and for even $\ell$ sets $\beta(j,\ell)=1-\beta(j,\ell-1)$. Then, \smash{$\widetilde{W}$} is defined as the distribution with density $\widetilde{w}$:
\begin{equation}
\label{eq:polya_sample_draw}
\widetilde{w}(x) = 2^J g_0^W(x) \prod_{j=1}^{J} \beta(j, k_j(x)),\;\;\;k_j(x) = \min\cb{ 2^j,  \lfloor 2^jG_0(x) +1 \rfloor},
\end{equation}
where $g_0^W$ is the density of the base distribution $G_0^W$. 

We set hyperparameters for the Dirichlet Process as in Section~\ref{sec:normal_means}. For the Pólya tree, we set $J=8$, $G_0^W = |\mathring{t}_8|$, where $\mathring{t}_8$ is the t-distribution with $8$ degrees of freedom standardized to have unit variance and $|\mathring{t}_8|$ is its folded version, and we set $\alpha(j, \ell) = 20j^2$ for $\ell=1,\dotsc,2^j-2$, $\alpha(j,\ell)=0.1$ for $\ell=2^j-1,2^j$. This choice is such that the outermost splits of the Pólya tree can remain as data-adaptive as possible.


\subsection{Computation}
\label{subsec:computation_unknown_shape}
We compute the partially Bayes p-values using an MCMC algorithm that jointly samples from the Dirichlet process posterior for the scale parameters and the Pólya tree posterior for the noise distribution. Our implementation extends Neal's Algorithm 8~\citep{neal2000markov} with multiple Metropolis-Hastings (MH) steps within each Gibbs update. 

We use a data augmentation step, in which we supplement $U_i$ with $\bar{Z}_i^{(b)}$, a new realization of the location estimate $\bar{Z}_i$ under the null hypothesis $\mu_i=0$. With this augmentation, we can reconstruct null datasets \smash{$\dataset_i^{(b)} = \{U_{i1} + \bar{Z}_i^{(b)}, \dotsc, U_{iK} + \bar{Z}_i^{(b)}\}$}. With a full dataset in hand, it is straightforward to carry out conjugate updates for the Pólya tree. We note that related augmentation strategies for sampling conditional on insufficient statistics are developed in~\citet{luciano2024insufficienta}.

Our sampler maintains the following state variables at iteration $b$:
\begin{itemize}[noitemsep,leftmargin=*]
\item cluster assignments $c_1^{(b)}, \dotsc, c_n^{(b)}$ where $c_i^{(b)} \in \{1,\dotsc,K^{(b)}\}$ and $K^{(b)}$ is the number of occupied clusters at iteration $b$;
\item cluster variance parameters $\sigma_1^{2(b)}, \dotsc, \sigma_{K^{(b)}}^{2(b)}$;
\item concentration parameter $c^{(b)}$ of the Dirichlet process;
\item P\'olya tree realization $W^{(b)}$ from $\mathrm{SymmPT}(\mathcal{A}, G_0^W, J)$;
\item null imputed test statistics $\bar{Z}_1^{(b)}, \dotsc, \bar{Z}_n^{(b)}$.
\end{itemize}
We provide details in Supplement~\ref{subsec:mcmc_polya}. Briefly, the first three bullets are analogous to state variables in Section~\ref{subsec:computation_normal_means}. Conceptually the main difference is that we can no longer rely on conjugate updates (as in Neal's Algorithm 2) and instead use MH-within-Gibbs as in Neal's Algorithm 8. In the fourth bullet we keep track of the symmetrized P\'olya tree realization $W^{(b)}$. 
The fifth bullet is the data augmentation step described above. After collecting $B$ post-burnin MCMC samples, we approximate the partially Bayes p-values using the imputed null test statistics:
\begin{equation}
\label{eq:pval_location}
\cpvalue_i \approx \frac{1 + \sum_{b=1}^B \mathbf{1}(|\bar{Z}_i^{(b)}| \geq |T_i|)}{1+B}.
\end{equation}

\section{Practical extension: Partially Bayes summarization}
\label{sec:partially_bayes_summarization}

For visualization and interpretation, it is useful to provide rejection regions for hypothesis $i$ of the form $\abs{T_i} \geq t_{\alpha}(V_i)$, where $V_i$ is a one-dimensional summary of the nuisance parameter statistic $U_i$ and $\alpha$ is the significance level. 
The rejection rule we have put forth so far in this paper is of the form $\cpvalue_i \leq \alpha$ and so it involves not only $T_i$ and $V_i$, but instead $T_i$ and $(U_1, \ldots, U_n)$.
We propose the following heuristic procedure:
\begin{enumerate}[leftmargin=*,noitemsep]
\item For each unit $i$, approximate the $\alpha$-level rejection threshold $t_i(U_1,\dotsc,U_n)$ by computing the $(1-\alpha)$-quantile of $\abs{T_i'}$ under the null, conditional on all $U_j$. This uses MCMC samples from computing $\cpvalue_i$.
\item Fit a nonparametric regression $t_i(U_1,\dotsc,U_n) \sim V_i$ for $i=1,\dotsc,n$, yielding $\hat{t}_{\alpha}(\cdot)$.
\item Approximate the decision boundary as $\abs{T_i} \geq \hat{t}_{\alpha}(V_i)$.
\end{enumerate}
We first present our rationale in the context of the normal means problem in Section~\ref{sec:normal_means} in which we summarize our data with nuisance parameter statistic $U_i = S_i^2$ in~\eqref{eq:normal_means_summary_stats}. Here we set $V_i=U_i$ (that is, $U_i$ is already one-dimensional and does not need to be summarized further).  In the empirical Bayes frame, in so far as $\cpvalue_i \approx \Pvalfunoracle(T_i, U_i, \Gstar),$ as in~\eqref{eq:approximation_n_to_infty}, our approach approximates the oracle decision boundary consisting of pairs $(t,u)$ such that $\Pvalfunoracle(t,u,\Gstar) = \alpha$. 

We next discuss our summarization procedure in the setting of Section~\ref{sec:unknown_shape}, where $U_i$ is the configuration statistic. By choosing $V_i = S_i^2$ (that is, the sample variance of the $U_{ij}$, or equivalently of the $Z_{ij}$), we aim to construct a decision boundary in the interpretable $(T_i, S_i^2)$ space, revealing how the test adapts to heteroscedasticity while learning the noise distribution nonparametrically. In the empirical Bayes frame, our summarization seeks to approximate the decision based on $\PP[\Gstar]{\abs{T'_i} \geq t \mid S_i^2=s^2}$. Notice that conditioning on $S_i^2$ comes with a loss of information compared to conditioning on $U_i$, however the upshot is interpretability (and potentially robustness to model misspecification as in e.g.,~\citet{doksum1990consistent, lewis2021bayesian, luciano2024insufficienta}---we do not pursue this angle here). Our approach may also be interpreted as approximating the decision boundary based on the following ``insufficient'' partially Bayes p-values (compare to~\eqref{eq:conditional_pred_pvalue}):
\begin{equation*}
  \begin{aligned}
  &\ipbpvalue_i := \ipbvaluefun_i(T_i, S_i^2, (U_j, j \neq i);\, \Pi),\;\;\\
  &\ipbvaluefun_i(t, s^2, (u_j, j \neq i);\, \Pi) := \Pi(|T_i'| \geq \abs{t} \,\mid \, S_i^2 = s^2,\, U_j = u_j,\, j \neq i),
  \end{aligned}
\end{equation*}
In words, we use $U_j, j \neq i$ to learn the noise distribution, but we only condition on $S_i^2$ (which provides information on the scale) for the unit for which we compute the p-value.

We conclude this section by noting that in the setting of Section~\ref{sec:unknown_shape}, our construction may also be interpreted as a new principled implementation of the conditional t suite of tests of~\citet{amaratunga2009conditional}. These authors also  construct rejection regions of the form $\abs{T_i} \geq t_{\alpha}(S_i^2)$, using a bootstrap approach (instead of posterior sampling) to learn the noise distribution.

\section{Simulation study}
\label{sec:numerical_study}

In this section we present a simulation study to evaluate the performance of our proposed partially Bayes p-values in the setting of Section~\ref{sec:unknown_shape} with unknown noise distribution shape.
We simulate according to model~\eqref{eq:location_unknown_shape} with $K=12$ and $n=10,000$. The $i$-th distribution $W_i$ is the Subbotin distribution with shape parameter $\xi>0$, scale parameter $b_i>0$ and location $\theta_i$ with density:
$$
w_i(z) := \frac{\xi}{2b_i\Gamma(1/\xi)}\exp\p{-\abs{(z-\theta_i)/b_i}^{\xi}},\;\;\; z \in \RR.
$$
Within a single simulation, the shape parameter $\xi$ is the same for all $i=1,\dotsc,n$, but we vary it across simulations as $\xi \in \{1,1.5,2,2.5,3\}$. The case $\xi=2$ corresponds to the normal distribution, i.e., to the setting of Section~\ref{sec:normal_means}, while $\xi=1$ corresponds to the Laplace distribution. As $\xi$ increases, the noise distribution $W_i$ becomes more light-tailed. Note that the variance of $W_i$ is equal to $\sigma_i^2 := \int u^2 W_i(\dd u)=b_i^2 \Gamma(3/\xi)/\Gamma(1/\xi)$. For comparability across different values of $\xi$, we specify $\sigma_i^2$ instead of $b_i$. We consider two settings: in the first setting, $\sigma_i^2=1$ for all $i$, and in the second setting the variances are heterogeneous and generated via \smash{$\sigma_i^2 \simiid \mathrm{Unif}[0.5,2]$}. Finally, in each simulation, we set $\mu_i=0$ for $90\%$ of the units (the nulls) and for the remaining $10\%$ we set $\theta_i=2.5\sigma_i / \sqrt{2}$ (the alternatives). Each simulation setting is repeated $100$ times and metrics are computed by averaging across the repetitions.

We compare four methods of constructing p-values:
\begin{enumerate}[noitemsep,leftmargin=*]
\item \textbf{t-test}: The standard t-test p-value, computed as $P_i = 2\bar{t}_{K-1}(\abs{T_i}/\sqrt{U_i})$, where $\bar{t}_{K-1}$ is the survival function of the t-distribution with $K-1$ degrees of freedom.
\item \textbf{Oracle}: The oracle p-value $\Pvalfunoracle(T_i, U_i, W_i)$ computed as in~\eqref{eq:config_pvalue}. We note that this p-value is not available in practice as it requires knowledge of both the noise distribution and its variance $\sigma_i^2$.
\item \textbf{Normal PB}: The partially Bayes p-values of Section~\ref{sec:normal_means} that assume normality of the noise distribution.
\item \textbf{P\'olya PB}: The partially Bayes p-values of Section~\ref{sec:unknown_shape} that use the Pólya tree prior for the noise distribution.
\end{enumerate}
We use these p-values in two ways, computing different metrics in each case:
\begin{enumerate}[leftmargin=*,noitemsep]
\item \textbf{Fixed-thresholding}: We reject all p-values $\leq 0.01$. We then report Monte Carlo estimates of $\PP{P_i \leq 0.01 \mid \theta_i=\theta_0}$ (which should be $0.01$ for uniform p-values) and the power of the fixed thresholding procedure, that is, the expected proportion of alternatives discovered: $\mathrm{Power}[P_i \leq 0.01] := \EE{\sum_{i=1}^n \ind(\theta_i \neq \theta_0,\, P_i \leq 0.01)/ \sum_{i=1}^n \ind(\theta_i \neq \theta_0)}$.
\item \textbf{BH}: We reject all p-values $\leq \hat{t}_{\text{BH}}$, where $\hat{t}_{\text{BH}}$ is the Benjamini-Hochberg threshold for controlling the false discovery rate (FDR) at level $0.1$. We report the FDR and power.
\end{enumerate}

\begin{figure}
  \centering
  \begin{tabular}{l}
  \footnotesize{a) Homoscedastic setting with $\sigma_i^2=1$ for all $i$}\\ 
  \includegraphics[width=\linewidth]{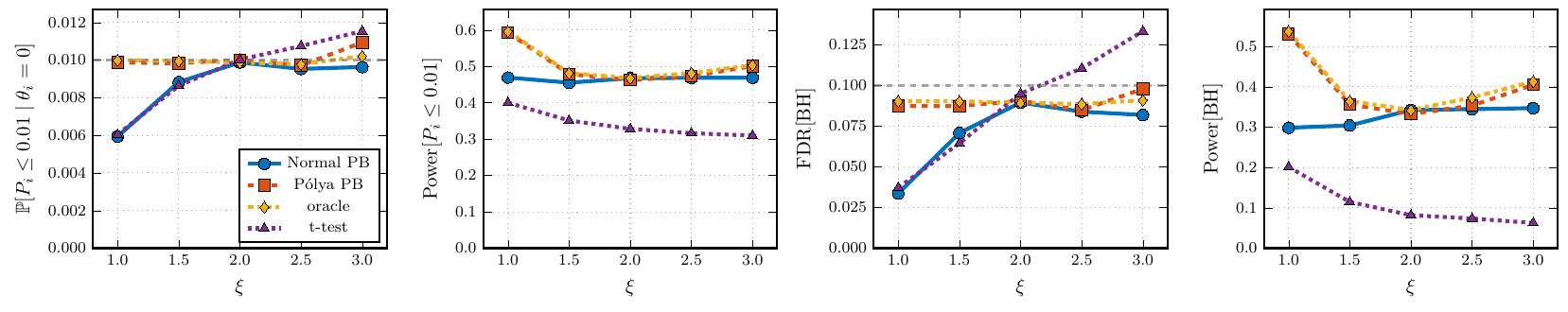} \\
  \footnotesize{b) Heteroscedastic setting with $\sigma_i^2 \sim \mathrm{Unif}[0.5, 2]$}\\
  \includegraphics[width=\linewidth]{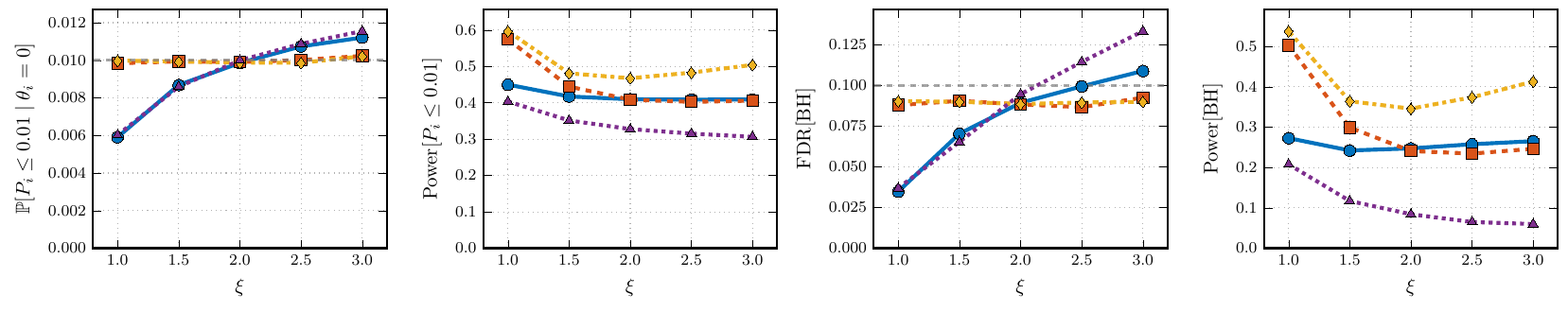}
  \end{tabular}
  \caption{Simulation results comparing the t-test, oracle p-values, and partially Bayes p-values under normality (Normal PB, ours) and with unknown distribution shape (P\'olya PB, ours) with a) homoscedasticity and b) heteroscedasticity. The x-axis shows the shape parameter $\xi$ of the Subbotin distribution. The columns show power and false discovery rate (FDR) for two rejection rules: fixed thresholding at 0.01 and the Benjamini-Hochberg procedure at 10\% FDR.}
  \label{fig:simulation_results}
\end{figure}

Results are shown in Fig.~\ref{fig:simulation_results}. We first discuss the homoscedastic case (panel a). Herein we see that the t-test does not provide type-I error control for $\xi \geq 2.5$ (that is, for more light-tailed noise distributions), while the other methods control type-I error with respect to both rejection rules and metrics. Normal PB is more powerful than the t-test for all $\xi$ and approaches the power of the oracle and P\'olya PB for $\xi = 2$ (in which case normality holds). Notably, P\'olya PB is nearly indistinguishable from the oracle for all $\xi$, demonstrating that it is possible to learn the shape of the noise distribution from the data and use it to construct powerful p-values. We next discuss the heteroscedastic case (panel b). The main differences here are as follows: first, Normal PB also loses type-I error for $\xi \geq 2.5$. Second, in this setting, there is a gap between the power of the oracle and P\'olya PB, since the oracle knows the true variances $\sigma_i^2$ exactly while P\'olya PB must account for uncertainty in the $\sigma_i^2$. (In the homoscedastic case, P\'olya PB is able to automatically learn that all the variances are identical, and so it can also learn all individual variances exactly.)
Still, P\'olya PB is powerful and outperforms other data-driven baselines.

\begin{figure}
\centering 
\includegraphics[width=\linewidth]{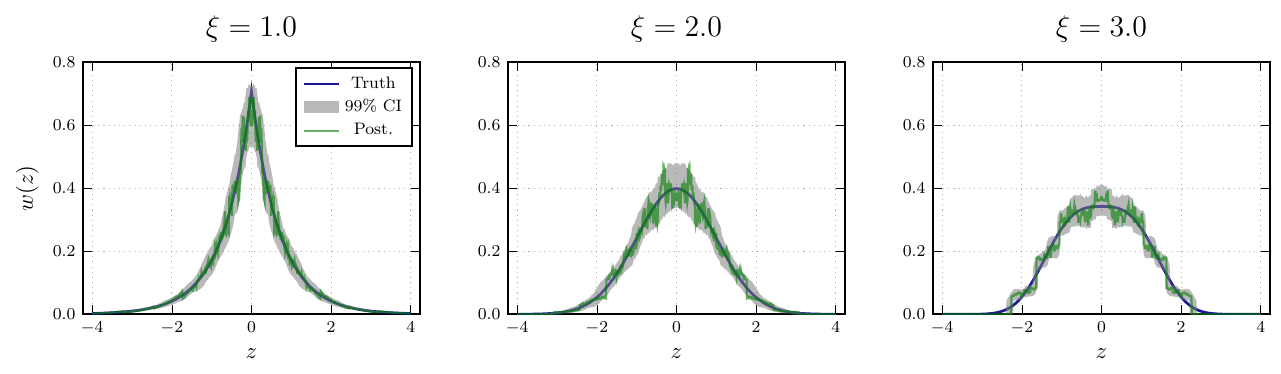}
\caption{Posterior samples of the noise distribution $W$ from the P\'olya tree prior,  normalized to variance $1$, for different values of the shape parameter $\xi$ of the Subbotin distribution in the heteroscedastic simulation. For each $\xi$ we plot the true density, 99\% pointwise credible intervals (CI) for the density (from a single simulation), as well as a the last posterior sample of the density from the MCMC algorithm.}
\label{fig:posterior_polya_trees}
\end{figure}

In Fig.~\ref{fig:posterior_polya_trees} we visualize the posterior of the centered noise distribution $W$ from the Pólya tree prior for different values of $\xi$ in the heteroscedastic setting. Specifically, for a single simulation repetition, we plot 99\% pointwise credible intervals of the true density, and we also plot the last posterior sample. We see that the Pólya tree is able to learn the shape of the noise distribution well, which explains the good performance of P\'olya PB in Fig.~\ref{fig:simulation_results}.

\section{Case studies}
\label{sec:real_data_application}

We illustrate our proposed methods on two real-world datasets.

\subsection{Revisiting differential gene expression in Crohn's disease}

We revisit the study of~\citet{palmieri2015genomewide} analyzed in Section
\ref{subsec:industrialist} on differential gene expression in Crohn's disease. We briefly recall that we considered three methods for computing p-values for each gene: the standard t-test, the partially Bayes p-values under normality (Normal PB), and the partially Bayes p-values with an unknown noise distribution shape (Pólya PB). Fig.~\ref{fig:palmieri} shows qq-plots comparing the quantiles of the three types of p-values against the uniform quantiles, as well as the number of p-values $\leq 0.001$ for each method. Both PB methods more than double the number of discoveries compared to the t-test. The two PB methods are similar, with Pólya PB yielding slightly more discoveries.

\begin{figure}
  \centering
  \begin{tabular}{ll}
  a)  & b) \\ 
  \includegraphics[width=0.63\linewidth]{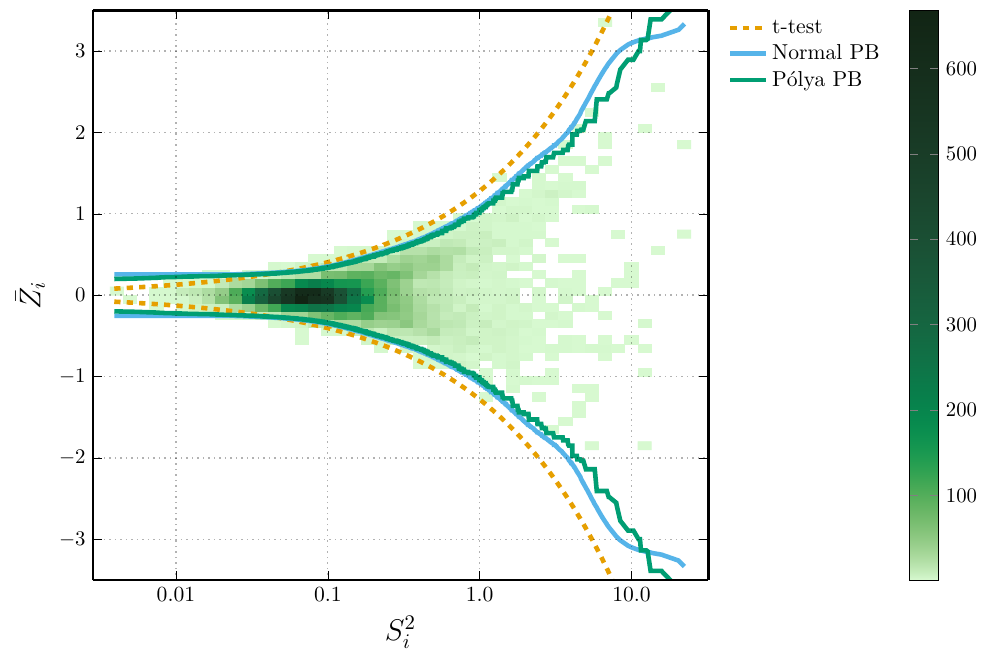}  & \includegraphics[width=0.33\linewidth]{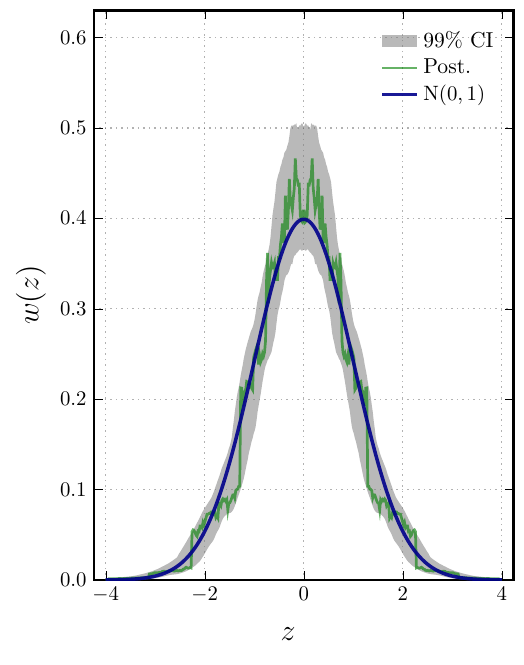} 
  \end{tabular}
  \caption{Continued reanalysis of the study of~\citet{palmieri2015genomewide}. a) Two-dimensional histogram of summary statistics $(\bar{Z}_i, S_i^2)$. The number of genes in each bin is indicated by the color. The curves indicate the decision boundary for rejecting at level $\alpha=0.001$ for the three methods. For the partially Bayes method, we use the summarization technique of Section~\ref{sec:partially_bayes_summarization}.
  b)
  Posterior samples of the noise distribution $W$ from the P\'olya tree prior, normalized to variance $1$. We show 99\% pointwise credible intervals (CI) for the density, with the last posterior sample from the MCMC algorithm and the standard normal density $\mathrm{N}(0,1)$ overlaid for reference.} 
  \label{fig:palmieri2}
\end{figure}

We supplement the p-value computation with two further visualizations. First, we apply the visualization strategy described in Section~\ref{sec:partially_bayes_summarization} to summarize the rejection regions of the three methods. In Fig.~\ref{fig:palmieri2}a), we plot the rejection regions at level $\alpha=0.001$ in the $(\bar{Z}_i, S_i^2)$ space for each method.  Compared to the standard t-test, the partially Bayes methods are more liberal for large values of $S_i^2$ and more conservative for small values of $S_i^2$. The reason is that both partially Bayes methods effectively borrow information across genes to learn the distribution of nuisance parameters and then appropriately shrink extreme values of $S_i^2$ toward more typical values. The Pólya PB rejection region is slightly larger that that of Normal PB. Second, in Fig.~\ref{fig:palmieri2}b), we visualize the posterior of the noise distribution $W$ using Pólya PB. The standard normal density is overlaid for reference and fully lies within the 99\% pointwise credible intervals. This suggests that the normality assumption is not contradicted for this dataset, although Pólya PB is still able to learn potential deviations from normality, e.g., typical posterior samples are more peaked than the normal density near the origin.

\subsection{Differences in differences with noisy control data}

\citet{gelman2021slamming} present a framework for paired experimental designs where systematic biases might influence outcomes. Specifically, they analyze data from an investigation into the health effects of low-frequency magnetic fields conducted in the 1980s by the U.S. Environmental Protection Agency~\citep{blackman1988influence}.
In each of $n=38$ experiments, chickens had their brains divided, with one half exposed to magnetic fields at a specific frequency (1--510 Hz) and the other serving as control. $Y_{i1}$ measured the average difference in calcium efflux between brain halves; a sham experiment with no magnetic field yielded $Y_{i0}$.
The model for the data of the $i$-th experiment is as follows:
\begin{equation}
Y_{i1} \sim \mathrm{N}(\mu_i+b_i,\, \sigma_{i1}^2),\;\;\; Y_{i0}  \sim \mathrm{N}(b_i,\, \sigma_{i0}^2),\;\;\; i=1,\dotsc,n.
\end{equation}
Above, $Y_{i1}$ and $Y_{i0}$ are the independent summarized outcomes from active and sham treatments respectively, $\mu_i$ represents the average treatment effect (the primary parameter) and $b_i$ denotes the systematic bias (the nuisance parameter). The variances $\sigma_{i1}^2$ and $\sigma_{i0}^2$ are assumed to be known. This model underlines the assumption that the bias $b_i$ affects both active and sham treatments in the same way and enables the identification of the actual treatment effect from the observed data since $\Delta Y_i = Y_{i1} - Y_{i0} \sim \mathrm{N}(\mu_i,\, \sigma_{i1}^2 + \sigma_{i0}^2)$, which no longer depends on $b_i$. P-values adjusting for the bias may then be computed as \smash{$P_i^{\text{sham}} = 2\Phi(-|\Delta Y_i| / \sigma_i)$} with \smash{$\sigma_i = (\sigma_{i1}^2 + \sigma_{i0}^2)^{1/2}$}.

\citet{gelman2021slamming} point out that, if $b_i=0$, then the adjustment above is needlessly conservative. It effectively doubles the variance of the test statistic from $\sigma_{i1}^2$ to $\sigma_{i1}^2 + \sigma_{i0}^2 \approx 2\sigma_{i1}^2$. One could instead compute p-values \smash{$P_i^{\mathrm{exp}} := 2\Phi(-\abs{Y_{i1}}/\sigma_{i1})$} ignoring sham treatments. Since $b_i$ is unknown, \citet{gelman2021slamming} propose a hierarchical Bayesian specification for $\mu_i, b_i$ that automatically determines how much to adjust for $b_i$. Inferences are then summarized through posterior means and credible intervals for $\mu_i$.

\begin{figure}
  \centering
  \begin{tabular}{ll}
    a) &
    b) \\
    \includegraphics[width=0.49\linewidth]{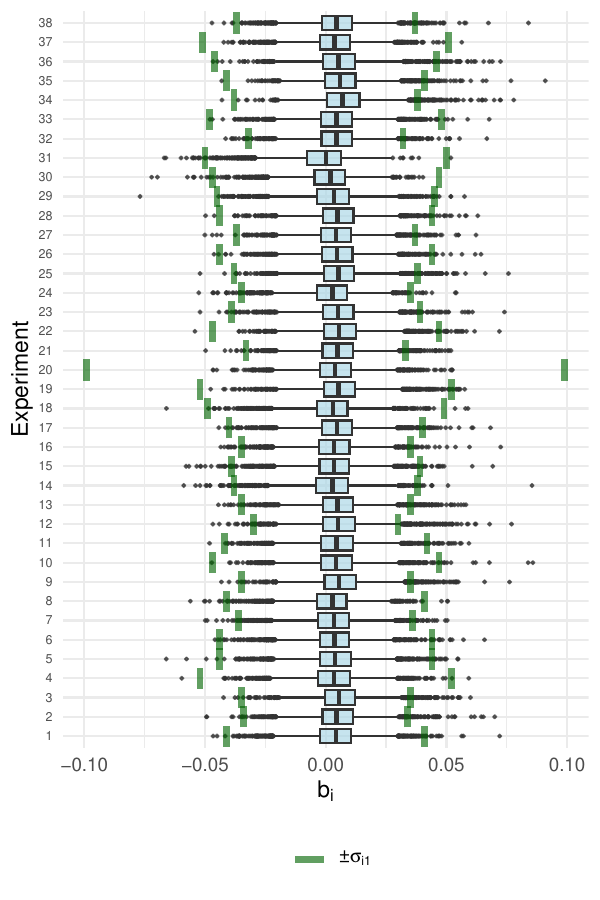} &
    \includegraphics[width=0.49\linewidth]{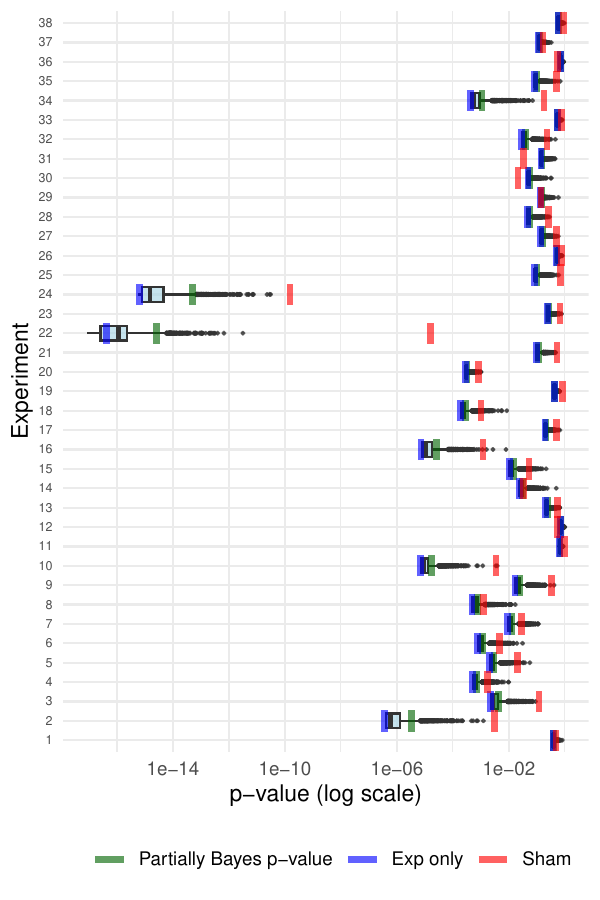} 
  \end{tabular}
  \caption{Illustration of partially Bayes methodology in the study of~\citet{blackman1988influence} with $n=38$ experiments. a) Boxplots of posterior draws for the nuisance parameters $b_i$ (systematic biases) for each experiment, with the known standard deviation $\sigma_{i1}$ of the treatment outcome shown for scale. b) Comparison of p-values: the sham-corrected p-values ($P_i^{\text{sham}}$), the experiment only p-values assuming no bias ($P_i^{\text{exp}}$), and our proposed partially Bayes p-values ($\cpvalue_i$). The boxplots for $\Pvalfunoracle(Y_{i1}, Y_{i2}, b_i)$ show the distribution of oracle p-values over the posterior draws of $b_i$, illustrating the uncertainty about the p-value due to the unknown bias.}
  \label{fig:boxplots_comparison}
\end{figure}

Here we explain how our framework with $T_i = Y_{i1}$, $U_i=Y_{i0}$, $\primary_i= \mu_i$, and $\nuisance_i=b_i$ can be used to form partially Bayes p-values. 
To apply our approach, given the small number of experiments, we pursue a parametric specification with improper uniform priors for the hyperparameters. Our hierarchical formulation for the nuisance parameters (\eqref{eq:hierarchy2} and~\eqref{eq:hierarchy3}) reads as
\smash{$b_i \cond \eta,\tau \, \sim \,  \mathrm{N}(\eta, \tau^2),\;
p(\eta, \tau) \propto 1$} and we fit this model using Stan~\citep{JSSv076i01}. \citet{gelman2021slamming} consider a specification in which no pooling occurs for $\mu_i$ but only for $b_i$. Our specification is very similar, but we treat $\mu_i$ in a frequentist fashion rather than placing an uninformative prior on it. The interpretation for $b_i$ is identical.

We show the results in Fig.~\ref{fig:boxplots_comparison}. Panel a) shows boxplots of 4,000 posterior draws for each nuisance parameter $b_i$ as well as $\sigma_{i1}$. As already noted by~\citet{gelman2021slamming}, most $b_i$ are small when compared to $\sigma_{i1}$, that is, the potential source of bias via $b_i$ is small when compared to the inherent noise in $Y_{i1}$. Panel b) shows boxplots of $\Pvalfunoracle(Y_{i1}, Y_{i2}, b_i)$, computed with the same $4,000$ posterior draws of $b_i$. The plot also shows the Monte-Carlo approximation of the partially Bayes p-values \smash{$\cpvalue_i$} (averaged over the $4,000$ posterior draws of $b_i$) as well as the sham-corrected p-values \smash{$P_i^{\text{sham}}$} and the p-values \smash{$P_i^{\text{exp}}$} that assume that $b_i=0$. We see that the p-values \smash{$P_i^{\text{sham}}$} are substantially larger than \smash{$P_i^{\text{exp}}$}, as expected. The partially Bayes p-values \smash{$\cpvalue_i$} are in between the two, and closer to \smash{$P_i^{\text{exp}}$} in most cases. This is because, as shown in panel a), most $b_i$ are small relative to $\sigma_{i1}$. Thus, the parametric partially Bayes p-values \smash{$\cpvalue_i$} automatically adapt to the data and avoid being overly conservative.

\section{Conclusion}

Partially Bayes p-values are a practical and general approach for large-scale inference that handles nuisance parameters by sharing information across units via Bayesian hierarchical modeling. Our proposal provides a principled Bayesian framework for the type of hybrid frequentist-(empirical) Bayesian p-values that practitioners already commonly use in high-throughput biology. We view it as a flexible compromise between Bayesian and frequentist philosophies: leveraging Bayesian nonparametrics to learn nuisance parameter distributions while providing approximate calibration guarantees that, in the spirit of~\citet{rubin1984bayesianly}, tie our  modeling to real-world frequency calculations.

\paragraph{Code availability.} In writing our initial prototype, we followed the code in the Particles.jl package~\citep{kleinschmidt2024particles}. All numerical results in this paper are fully third party reproducible using code available on Github (\url{https://github.com/nignatiadis/partially-bayes-pvalues-paper}).

\paragraph{Acknowledgements.} We thank Surya Tokdar for helpful discussions. This work was completed in part with resources provided by the University of Chicago’s Research Computing Center.
The authors gratefully acknowledge support from the U.S. National Science Foundation (DMS-2443410 for NI and DMS-2152999 for LM).

\bibliographystyle{abbrvnat}
\bibliography{conditional_predictive}

\appendix

\setcounter{equation}{0}
\setcounter{figure}{0}
\setcounter{table}{0}
\setcounter{prop}{0}

\renewcommand{\theequation}{S\arabic{equation}}
\renewcommand{\thefigure}{S\arabic{figure}}
\renewcommand{\thetable}{S\arabic{table}}
\renewcommand{\theprop}{S\arabic{prop}}
\renewcommand{\thetheo}{S\arabic{theo}}
\renewcommand{\thelemm}{S\arabic{lemm}}

\section{Further theoretical results}
\label{sec:further_theoretical_results}

\begin{prop}[Calibration as $K\to \infty$ in the empirical Bayes frame]
\label{prop:K_to_infty_empirical_bayes}
Let $\nuisancetrue_i \sim \Gstar$. Suppose that Assumptions~\ref{assu:test_statistics} and \ref{assu:regular_pvalues} hold, and that for all $i \in \Hnull$ we have that:
\begin{itemize}
\item[$(*')$] For any $\delta >0$, 
$\,\EE[\Gstar]{\EE[\nuisancevectortrue]{ \Pi(\nu_i\,:\, d(\nuisance_i, \nuisance_i^{\star}) > \delta \mid U_1,\ldots,U_n)}} \to 0\, \text{ as }\, K \to \infty.$    
\end{itemize}
Then, as $K \to \infty$ (with $n$ fixed),
$$ \limsup_{K \to \infty} \max_{i \in \Hnull} \sup_{\alpha \in [0,1]}\abs{\PP[\Gstar, \primaryvectortrue]{ \cpvalue_i \leq \alpha} - \alpha} =0.$$
\end{prop}
Note that the condition $(*')$ of this proposition is entirely analogous to condition $(*)$ in Theorem~\ref{theo:K_to_infty}, the only difference being that the expectation is now also taken with respect to the nuisance parameter distribution $\Gstar$. The proof is analogous to the proof of Theorem~\ref{theo:K_to_infty} in Supplement~\ref{subsec:proof_K_to_infty}.

\section{Proofs}
\subsection{Auxiliary lemmata}

\begin{lemm}
  \label{lemm:deterministic_indicator_to_l1}
Let $P_i, P_i^* \in [0,1]$. For any $\delta \in (0, 1)$ and $\alpha \in [0, 1]$, it holds that:
$$ \ind(P_i  \leq \alpha) - \ind(P_i^* \leq \alpha+\delta) \leq \frac{1}{\delta}\abs{P_i - P_i^*}.$$
Similarly, 
$$ \ind(P_i  \leq \alpha) - \ind(P_i^* \leq \alpha-\delta) \geq -\frac{1}{\delta}\abs{P_i - P_i^*}.$$
\end{lemm}

\begin{proof}
For the first inequality, we note that the left-hand side is positive only if $P_i \leq \alpha$ and $P_i^* > \alpha+\delta$. In this case, the inequality holds because:
$$
1 \leq \frac{P_i^* - P_i}{\delta} = \frac{\abs{P_i - P_i^*}}{\delta}.
$$
For the second inequality, the left-hand side is negative only if $P_i > \alpha$ and $P_i^* \leq \alpha-\delta$. In this case, the inequality holds because:
$$
-1 \geq -\frac{P_i - P_i^*}{\delta} = -\frac{\abs{P_i - P_i^*}}{\delta}.
$$
\end{proof}

\begin{lemm}
\label{lemm:l1_to_uniformity}
Consider pairs of $[0,1]$-valued random variables $(P_1^K, Q_i^K), \ldots, (P_n^K, Q_n^K)$ where both $K$ and $n$ are positive integers. Moreover write $K=K(m)$ and $n=n(m)$ for some $m \in \NN$ that indexes asymptotics in which $K$ or $n$ may grow. Let $\mathcal{H}_0 \equiv \mathcal{H}_0^m \subset \cb{1,\ldots,n}$ be a subset of indices.
We have the following results.
\begin{enumerate}[leftmargin=*,label=(\alph*)]
\item Suppose that $Q_i^K$ is uniformly distributed on $[0,1]$ for all $i \in \mathcal{H}_0$, that is, $\PP{Q_i^K \leq \alpha} = \alpha$ for all $\alpha \in [0,1]$. Suppose also that
$$
\max_{i \in \mathcal{H}_0} \cb{\EE{\abs{P_i^K - Q_i^K}}} \to 0 \, \text{ as }\, m \to \infty.
$$
Then, 
it follows that:
$$
\max_{i \in \mathcal{H}_0} \sup_{\alpha \in [0,1]}\abs{\PP{P_i^K \leq \alpha} - \alpha} \to 0\, \text{ as }\, m \to \infty.
$$
\item Suppose that $Q_1^K,\ldots,Q_n^K$ are compound p-values, that is, suppose that
$$
\frac{1}{n}\sum_{i \in \mathcal{H}_0} \PP{Q_i^K \leq \alpha} \leq \alpha\, \text{ for all }\, \alpha \in [0,1],
$$
and also suppose that
$$
\frac{1}{n}\sum_{i \in \mathcal{H}_0} \EE{\abs{P_i^K - Q_i^K}} \to 0 \, \text{ as }\, m \to \infty.
$$
Then it follows that $P_1^K,\ldots,P_n^K$ are asymptotically compound p-values, i.e.,
$$\limsup_{m \to \infty} \sup_{\alpha \in [0,1]}\bigg(\frac{1}{n}\sum_{i \in \mathcal{H}_0} \PP{P_i^K \leq \alpha} - \alpha\bigg)_+ =\,0,$$
where $a_+ = \max\{0,a\}$ for $a \in \RR$.
\end{enumerate}
\end{lemm}
\begin{proof}
We first prove part (a). Let $\varepsilon > 0$ be arbitrary. By Lemma~\ref{lemm:deterministic_indicator_to_l1}, for any $\delta \in (0,1)$, $\alpha \in [0,1]$, and $i \in \mathcal{H}_0$:
$$\ind(P_i^K \leq \alpha) \leq \ind(Q_i^K \leq \alpha + \delta) + \frac{1}{\delta}\abs{P_i^K - Q_i^K}$$
and
$$\ind(P_i^K \leq \alpha) \geq \ind(Q_i^K \leq \alpha - \delta) - \frac{1}{\delta}\abs{P_i^K - Q_i^K}.$$
Taking expectations for a fixed $i \in \mathcal{H}_0$:
$$\PP{P_i^K \leq \alpha} \leq \PP{Q_i^K \leq \alpha + \delta} + \frac{1}{\delta}\EE{\abs{P_i^K - Q_i^K}}$$
and
$$\PP{P_i^K \leq \alpha} \geq \PP{Q_i^K \leq \alpha - \delta} - \frac{1}{\delta}\EE{\abs{P_i^K - Q_i^K}}.$$
Since $Q_i^K \sim \text{Unif}[0,1]$, we have $\PP{Q_i^K \leq x} = x$ for $x \in [0,1]$. This implies $\PP{Q_i^K \leq \alpha + \delta} \leq \alpha + \delta$ and $\PP{Q_i^K \leq \alpha - \delta} \geq \alpha - \delta$. Thus:
$$\alpha - \delta - \frac{1}{\delta}\EE{\abs{P_i^K - Q_i^K}} \leq \PP{P_i^K \leq \alpha} \leq \alpha + \delta + \frac{1}{\delta}\EE{\abs{P_i^K - Q_i^K}}.$$
Therefore, for each $i \in \mathcal{H}_0$ and any $\alpha \in [0,1]$:
$$\abs{\PP{P_i^K \leq \alpha} - \alpha} \leq \delta + \frac{1}{\delta}\EE{\abs{P_i^K - Q_i^K}}.$$
Let $\eta_m := \max_{i \in \mathcal{H}_0} \EE{\abs{P_i^K - Q_i^K}}$. By assumption, $\eta_m \to 0$ as $m \to \infty$.
The above inequality implies:
$$\max_{i \in \mathcal{H}_0} \sup_{\alpha \in [0,1]}\abs{\PP{P_i^K \leq \alpha} - \alpha} \leq \delta + \frac{1}{\delta}\eta_m.$$
Choose $\delta = \sqrt{\eta_m}$. Then:
$$\max_{i \in \mathcal{H}_0} \sup_{\alpha \in [0,1]}\abs{\PP{P_i^K \leq \alpha} - \alpha} \leq \sqrt{\eta_m} + \frac{1}{\sqrt{\eta_m}}\eta_m = 2\sqrt{\eta_m}.$$
Since $\eta_m \to 0$ as $m \to \infty$, the result for part (a) follows.

Now we prove part (b). From the first inequality derived in the proof of part (a), we have for any $i \in \mathcal{H}_0$:
$$\PP{P_i^K \leq \alpha} \leq \PP{Q_i^K \leq \alpha + \delta} + \frac{1}{\delta}\EE{\abs{P_i^K - Q_i^K}}.$$
Summing over $i \in \mathcal{H}_0$ and dividing by $n$:
$$\frac{1}{n}\sum_{i \in \mathcal{H}_0} \PP{P_i^K \leq \alpha} \leq \frac{1}{n}\sum_{i \in \mathcal{H}_0} \PP{Q_i^K \leq \alpha + \delta} + \frac{1}{n}\sum_{i \in \mathcal{H}_0} \frac{1}{\delta}\EE{\abs{P_i^K - Q_i^K}}.$$
By the assumption that $Q_i^K$ are compound p-values, we have $\frac{1}{n}\sum_{i \in \mathcal{H}_0} \PP{Q_i^K \leq \alpha + \delta} \leq \alpha + \delta$.
Let $\eta_m := \frac{1}{n}\sum_{i \in \mathcal{H}_0} \EE{\abs{P_i^K - Q_i^K}}$. By assumption, $\eta_m \to 0$ as $m \to \infty$. Then:
$$\frac{1}{n}\sum_{i \in \mathcal{H}_0} \PP{P_i^K \leq \alpha} \leq \alpha + \delta + \frac{1}{\delta}\eta_m.$$
This implies:
$$\frac{1}{n}\sum_{i \in \mathcal{H}_0} \PP{P_i^K \leq \alpha} - \alpha \leq \delta + \frac{1}{\delta}\eta_m.$$
Since the right-hand side is always positive, taking the positive part of the left-hand side and then the supremum over $\alpha \in [0,1]$ gives:
$$\sup_{\alpha \in [0,1]}\left(\frac{1}{n}\sum_{i \in \mathcal{H}_0} \PP{P_i^K \leq \alpha} - \alpha\right)_+ \leq \delta + \frac{1}{\delta}\eta_m.$$
Choosing $\delta = \sqrt{\eta_m}$ yields the bound $2\sqrt{\eta_m}$. Since $\eta_m \to 0$ as $m \to \infty$, we conclude the proof of part (b).
\end{proof}

\begin{lemm}
\label{lem:pvalue_ratios_to_diffs}
For any distributions $G,H$ supported on $\nuisancespace$ and any $u$ such that $f(u;G)>0$, it holds that:
$$\abs{\Pvalfunoracle(t, u, G) - \Pvalfunoracle(t, u, H)} \leq 2\abs{\frac{N(t,u;G) - N(t,u;H)}{f(u;G)}} \,+\,2\abs{\frac{f(u;G) - f(u;H)}{f(u;G)}},$$
where $f(u;G)$ is defined in~\eqref{eq:marginal_density} and $N(t,u;G)$ is defined as follows:
\begin{equation}
N(t,u;G) := \int \Pvalfunoracle(t,u,\nu) p(u \mid \nu) \, G(\dd\nu).
\label{eq:N_function}
\end{equation}
\end{lemm}
\begin{proof}
We first note the following equality:
$$ 
\Pvalfunoracle(t, u, G) = \frac{\int \Pvalfunoracle(t,u,\nu) p(u \mid \nu) \, G(\dd\nu)}{\int p(u \mid \nu) \, G(\dd\nu)} = \frac{N(t,u;G)}{f(u;G)},
$$
as long as $f(u;G) > 0$. The same holds for $H$.

Now, let $N_G := N(t,u;G)$, $f_G := f(u;G)$, $N_H := N(t,u;H)$, and $f_H := f(u;H)$. We want to bound $\abs{N_G/f_G - N_H/f_H}$. Let $M:= (G+H)/2$. By linearity of integration, $N(t,u;M) = (N_G+N_H)/2$. Let's call this $N_M$. Similarly, let $f_M := f(u;M) = (f_G + f_H)/2$. 
We use the triangle inequality by adding and subtracting intermediate terms:
\begin{align*}
\abs{\frac{N_G}{f_G} - \frac{N_H}{f_H}} &= \abs{ \frac{N_G}{f_G} - \frac{N_G}{f_M} + \frac{N_G}{f_M} - \frac{N_H}{f_M} + \frac{N_H}{f_M} - \frac{N_H}{f_H} } \\
&\leq \abs{N_G} \abs{\frac{1}{f_G} - \frac{1}{f_M}} + \frac{\abs{N_G - N_H}}{\abs{f_M}} + \abs{N_H} \abs{\frac{1}{f_M} - \frac{1}{f_H}} \\
&= \frac{\abs{N_G}}{\abs{f_G}} \frac{\abs{f_M-f_G}}{\abs{f_M}} + \frac{\abs{N_G - N_H}}{\abs{f_M}} + \frac{\abs{N_H}}{\abs{f_H}} \frac{\abs{f_H-f_M}}{\abs{f_M}}.
\end{align*}
By definition, $\Pvalfunoracle(t,u,G),\Pvalfunoracle(t,u,H) \in [0,1]$, so $\abs{N_G/f_G} \leq 1$ and $\abs{N_H/f_H} \leq 1$.
The inequality becomes:
\begin{align*}
\abs{\frac{N_G}{f_G} - \frac{N_H}{f_H}} &\leq \frac{\abs{f_M-f_G}}{\abs{f_M}} + \frac{\abs{N_G - N_H}}{\abs{f_M}} + \frac{\abs{f_H-f_M}}{\abs{f_M}}.
\end{align*}
Using $f_M-f_G = (f_H-f_G)/2$ and $f_H-f_M = (f_H-f_G)/2$, we get:
\begin{align*}
\abs{\frac{N_G}{f_G} - \frac{N_H}{f_H}} &\leq \frac{2\abs{f(u;G)-f(u;H)} + 2\abs{N(t,u;G) - N(t,u;H)}}{f(u;G)+f(u;H)}.
\end{align*}
Since $f(u;H) \geq 0$, we have the the desired result:
$$
\abs{\Pvalfunoracle(t, u, G) - \Pvalfunoracle(t, u, H)} \leq 2\frac{\abs{N(t,u;G) - N(t,u;H)}}{f(u;G)} + 2\frac{\abs{f(u;G) - f(u;H)}}{f(u;G)}.
$$
\end{proof}

\subsection{Proof of Theorem~\ref{theo:K_to_infty}}
\label{subsec:proof_K_to_infty}
\begin{proof}
Fix $i \in \Hnull$. Let $\cpvalue_i = \cpvaluefun_i(T_i^K, (U_1^K, \ldots, U_n^K); \Pi)$ be the partially Bayes p-value and let $P_i^{\star,K} = \Pvalfunoracle[K](T_i^K, U_i^K, \nuisancetrue_i)$ be the oracle p-value with access to the true nuisance parameter. By Proposition~\ref{prop:oracle_nuisance_pvalue}, $P_i^{\star,K} \sim \mathrm{Unif}[0,1]$.

By Theorem~\ref{theo:oracle_representation}(a), we have
\begin{align*}
\cpvalue_i &= \EE[\Pi]{\Pvalfunoracle[K](T_i^K, U_i^K, \nu_i) \mid U_1^K, \ldots, U_n^K}.
\end{align*}
By Jensen's inequality, we have:
\begin{align*}
\EE[\nuisancevectortrue]{\abs{\cpvalue_i - P_i^{\star,K}}} &= \EE[\nuisancevectortrue]{\abs{\EE[\Pi]{\Pvalfunoracle[K](T_i^K, U_i^K, \nu_i) \mid U_1^K, \ldots, U_n^K} - \Pvalfunoracle[K](T_i^K, U_i^K, \nuisancetrue_i)}} \\
&\leq \EE[\nuisancevectortrue]{\EE[\Pi]{\abs{\Pvalfunoracle[K](T_i^K, U_i^K, \nu_i) - \Pvalfunoracle[K](T_i^K, U_i^K, \nuisancetrue_i)} \mid U_1^K, \ldots, U_n^K}}.
\end{align*}
For any $\varepsilon > 0$, by Assumption~\ref{assu:regular_pvalues}, there exists $\delta > 0$ such that if $d(\nu, \nuisancetrue_i) \leq \delta$, then $\abs{\Pvalfunoracle[K](T_i^K, U_i^K, \nu) - \Pvalfunoracle[K](T_i^K, U_i^K, \nuisancetrue_i)} < \varepsilon$. Therefore:
\begin{align*}
\EE[\Pi]{\abs{\Pvalfunoracle[K](T_i^K, U_i^K, \nu_i) - \Pvalfunoracle[K](T_i^K, U_i^K, \nuisancetrue_i)} \mid U_1^K, \ldots, U_n^K} \\
\leq \varepsilon + \Pi(d(\nu_i, \nuisancetrue_i) > \delta \mid U_1^K, \ldots, U_n^K).
\end{align*}
Taking expectations with respect to $\nuisancevectortrue$ and using assumption $(*)$:
\begin{align*}
\EE[\nuisancevectortrue]{\abs{\cpvalue_i - P_i^{\star,K}}} &\leq \varepsilon + \EE[\nuisancevectortrue]{\Pi(d(\nu_i, \nuisancetrue_i) > \delta \mid U_1^K, \ldots, U_n^K)} \to \varepsilon\, \text{ as }\, K \to \infty.
\end{align*}
Since $\varepsilon$ was arbitrary, we have:
$$\max_{i \in \Hnull}\cb{\EE[\nuisancevectortrue]{\abs{\cpvalue_i - P_i^{\star,K}}}} \to 0\, \text{ as }\,K \to \infty.$$
The result follows by applying Lemma~\ref{lemm:l1_to_uniformity} with $P_i^K = \cpvalue_i$ and $Q_i^K = P_i^{\star,K}$.
\end{proof}

\subsection{Proof of Theorem~\ref{theo:calibration_n_to_infty_empirical_bayes}}
\label{subsec:proof_calibration_n_to_infty_empirical_bayes}

\begin{proof}
We will first show that under the above conditions:
\begin{equation}
\label{eq:l1_consistency_cpvalue}
\max_{i \in \Hnull}\EE[\Gstar]{ \abs{\cpvalue_i - \Pvalfunoracle(T_i, U_i, \Gstar)}} \to 0 \,\text{ as }\, n \to \infty.
\end{equation}
Restated, this is saying that,
$$
\max_{i \in \Hnull}\EE[\Gstar]{ \abs{\cpvaluefun_i(T_i, (U_1,\dotsc,U_n);\, \Pi)- \Pvalfunoracle(T_i, U_i, \Gstar)}} \to 0 \,\text{ as }\, n \to \infty.
$$
Fix $i \in \Hnull$. Let $\varepsilon > 0$ be arbitrary. By Assumption~\ref{assu:nuisance_statistics}, the collection of distributions $\{\mathbb P[U_i \in \cdot \mid \nuisance_i]: \nuisance_i \in \nuisancespace\}$ is tight. Therefore, there exists a compact set $C \subset \mathcal{U}$ such that
By tightness, for any $\varepsilon > 0$, there exists a compact set $C \subset \mathcal{U}$ such that
$\sup_{\nuisance_i \in \nuisancespace} \PP[\nuisance_i]{U_i \notin C} < \varepsilon$ and thus also such that $\PP[\Gstar]{U_i \notin C} < \varepsilon$.
We can then write
\begin{align*}
&\EE[\Gstar]{ \abs{\cpvaluefun_i(T_i, (U_1,\dotsc,U_n);\, \Pi)- \Pvalfunoracle(T_i, U_i, \Gstar)} }\\
&\;\;\;\; \leq \EE[\Gstar]{ \abs{\cpvaluefun_i(T_i, (U_1,\dotsc,U_n);\, \Pi)- \Pvalfunoracle(T_i, U_i, \Gstar)} \ind(U_i \in C) } \,+ \,\PP[\Gstar]{ \ind(U_i \notin C) } \\
&\;\;\;\;\leq \EE[\Gstar]{ \abs{\cpvaluefun_i(T_i, (U_1,\dotsc,U_n);\, \Pi)- \Pvalfunoracle(T_i, U_i, \Gstar)} \ind(U_i \in C}
 \,+\, \varepsilon,
\end{align*}
where we used that both p-values are in $[0,1]$ so their difference is at most $1$.

Let us use the hand notation $G_{\Pi}^{-i} := G_{\Pi}[U_1,\ldots,U_{i-1},U_{i+1}, \ldots,U_n]$. Then, by Theorem~\ref{theo:oracle_representation}b')
we have that $\cpvaluefun_i(t, (U_1,\dotsc,U_n);\, \Pi) = \Pvalfunoracle(t, U_i, G_{\Pi}^{-i})$. Using Lemma~\ref{lem:pvalue_ratios_to_diffs} and the notation therein, we have that:
$$
\begin{aligned}
&\abs{\Pvalfunoracle(t, U_i, G_{\Pi}^{-i})  - \Pvalfunoracle(t, U_i, \Gstar)} \\ 
&\;\;\;\;\;\leq  2\abs{\frac{N(t,U_i;G_{\Pi}^{-i}) - N(t,U_i;\Gstar)}{f(U_i;\Gstar)}} \,+\,2\abs{\frac{f(U_i;G_{\Pi}^{-i}) - f(U_i;\Gstar)}{f(U_i;\Gstar)}}.
\end{aligned}
$$
Let us bound $\sup_{u \in C}\abs{f(u;G_{\Pi}^{-i}) - f(u;\Gstar)}$. By Assumption~\ref{assu:nuisance_statistics}, the functions $\{\nu \mapsto p(u \mid \nu) : u \in C\}$ are uniformly bounded and uniformly equicontinuous on the compact set $\nuisancespace$. By the Arzelà-Ascoli theorem, this collection is relatively compact in the space of continuous functions on $\nuisancespace$ equipped with the supremum norm. Since any continuous function on a compact metric space can be approximated arbitrarily well by Lipschitz functions (by standard results in approximation theory), for any $\eta > 0$, there exists a finite collection of bounded Lipschitz functions $\{\psi_1, \ldots, \psi_M\}$ such that for any $u \in C$:
$$
\sup_{\nu \in \nuisancespace} \abs{p(u \mid \nu) - \psi_{j(u)}(\nu)} < \eta
$$
for some $j(u) \in \{1, \ldots, M\}$.
Now, for any $u \in C$:
\begin{align*}
\abs{f(u;G_{\Pi}^{-i}) - f(u;\Gstar)} &= \abs{\int p(u \mid \nu) [G_{\Pi}^{-i} - \Gstar](\dd\nu)} \\
&\leq \abs{\int \psi_{j(u)}(\nu) [G_{\Pi}^{-i} - \Gstar](\dd\nu)} + 2\eta \\
&\leq \Norm{\psi_{j(u)}}_{\text{Lip}} \Dbl(G_{\Pi}^{-i}, \Gstar) + 2\eta,
\end{align*}
where 
$$\Norm{\psi_j}_{\text{Lip}} := \Norm{\psi_j}_\infty + \sup_{\nu \neq \nu'} \frac{|\psi_j(\nu) - \psi_j(\nu')|}{d(\nu,\nu')}$$
denotes the Lipschitz norm. Taking the supremum over $u \in C$ and using the fact that there are only finitely many functions $\psi_j$:
$$\sup_{u \in C}\abs{f(u;G_{\Pi}^{-i}) - f(u;\Gstar)} \leq \max_{j=1,\ldots,M} \Norm{\psi_j}_{\text{Lip}} \Dbl(G_{\Pi}^{-i}, \Gstar) + 2\eta.$$
For $N(t,u;G) := \int \Pvalfunoracle(t,u,\nu) p(u \mid \nu) \, G(\dd\nu)$, similar to the analysis of $f(u;G)$, we need to control $\sup_{t \in \mathbb{R}, u \in C}\abs{N(t,u;G_{\Pi}^{-i}) - N(t,u;\Gstar)}$.
By Assumptions~\ref{assu:regular_pvalues} and~\ref{assu:nuisance_statistics}, the functions $\{\nu \mapsto \Pvalfunoracle(t,u,\nu) p(u \mid \nu) : t \in \mathbb{R}, u \in C\}$ are uniformly bounded and uniformly equicontinuous on $\nuisancespace$. By the Arzelà-Ascoli, for any $\eta > 0$, there exists a finite collection of bounded Lipschitz functions $\{\psi_{j}^* : j = 1, \ldots, L\}$ such that for any $(t,u) \in \mathbb{R} \times C$:
$$
\sup_{\nu \in \nuisancespace} \abs{\Pvalfunoracle(t,u,\nu) p(u \mid \nu) - \psi_{j(t,u)}^*(\nu)} < \eta
$$
for some $j(t,u) \in \{1, \ldots, L\}$.
Now, for any $(t,u) \in \mathbb{R} \times C$:
\begin{align*}
\abs{N(t,u;G_{\Pi}^{-i}) - N(t,u;\Gstar)} &= \abs{\int \Pvalfunoracle(t,u,\nu) p(u \mid \nu) [G_{\Pi}^{-i} - \Gstar](\dd\nu)} \\
&\leq \abs{\int \psi_{j(t,u)}^*(\nu) [G_{\Pi}^{-i} - \Gstar](\dd\nu)} + 2\eta \\
&\leq \Norm{\psi_{j(t,u)}^*}_{\text{Lip}} \Dbl(G_{\Pi}^{-i}, \Gstar) + 2\eta.
\end{align*}
Taking the supremum over $(t,u) \in \mathbb{R} \times C$ and using the finite collection:
$$\sup_{t \in \mathbb{R}, u \in C}\abs{N(t,u;G_{\Pi}^{-i}) - N(t,u;\Gstar)} \leq \max_{j=1,\ldots,L} \Norm{\psi_j^*}_{\text{Lip}} \Dbl(G_{\Pi}^{-i}, \Gstar) + 2\eta.$$
Now define $B := \max_{j=1,\ldots,L} \Norm{\psi_j^*}_{\text{Lip}} + \max_{j=1,\ldots,M} \Norm{\psi_j}_{\text{Lip}}$. Then we have:
\begin{equation}
\abs{\Pvalfunoracle(T_i, U_i, G_{\Pi}^{-i}) - \Pvalfunoracle(T_i, U_i, \Gstar)}\ind(U_i \in C) \leq \frac{2 B \Dbl(G_{\Pi}^{-i}, \Gstar)  + 4\eta}{f(U_i;\Gstar)}\ind(U_i \in C)
\label{eq:bound_on_diff_pvalues_on_C}
\end{equation}
with the convention $1/0 = \infty$. Next, note that,
$$
\EE[\Gstar]{\frac{1}{f(U_i;\Gstar)}\ind(U_i \in C)} = \int_C \frac{1}{f(u;\Gstar)}f(u;\Gstar) \ind(f(u;\Gstar) >0)\, \dd\lambda(u) \leq \lambda(C) < \infty.
$$
Also observe that $\Dbl(G_{\Pi}^{-i}, \Gstar)$ is a function of $U_1,\ldots,U_{i-1},U_{i+1},\ldots,U_n$ and thus independent of $U_i$. Therefore, combining all results so far, we have that:
\begin{equation}
\begin{aligned}
&\EE[\Gstar]{ \abs{\cpvaluefun_i(T_i, (U_1,\dotsc,U_n);\, \Pi)- \Pvalfunoracle(T_i, U_i, \Gstar)} } \\
&\;\;\;\;\;\leq 4B \EE[\Gstar]{\Dbl(G_{\Pi}^{-i}, \Gstar)} \lambda(C) + 4\eta \lambda(C) + \varepsilon.
\end{aligned}
\label{eq:main_bound_cpvalue_difference}
\end{equation}
Now notice that $\EE[\Gstar]{\Dbl(G_{\Pi}^{-i}, \Gstar)}=\EE[\Gstar]{\Dbl(G_{\Pi}[U_1,\dots,U_{n-1}], \Gstar)}$ by exchangeability. Thus, by $(*)$ of the theorem statement, we have that $\max_{i \in \Hnull}\EE[\Gstar]{\Dbl(G_{\Pi}^{-i}, \Gstar)} \to 0$ as $n \to \infty$. Hence, \eqref{eq:l1_consistency_cpvalue} follows by (i) taking $n \to \infty$, (ii) then taking $\eta \to 0$, and (iii) finally taking $\varepsilon \to 0$.

Finally, we apply Lemma~\ref{lemm:l1_to_uniformity} with $P_i^K = \cpvaluefun_i(T_i, (U_1,\dotsc,U_n);\, \Pi)$ and $Q_i^K = \Pvalfunoracle(T_i, U_i, \Gstar)$ to conclude with the proof the theorem.

\end{proof}

\subsection{Proof of Proposition~\ref{prop:TV_to_BL}}

\begin{proof}
The proof proceeds in three steps. First, we establish that the map from the nuisance parameter distribution $G$ to the marginal density of the nuisance statistic $f(\cdot; G)$ is continuous. Second, we use this continuity and the identifiability assumption to show that if $G$ is far from $\Gstar$, then $f(\cdot; G)$ must also be far from $f(\cdot; \Gstar)$ in a uniform way. Finally, the assumed posterior consistency for the marginal density will imply posterior consistency for the nuisance parameter distribution.\\

\noindent \textbf{Step 1: Continuity.} Consider a sequence of distributions $G_m$ on $\nuisancespace$ indexed by $m \in \NN$. Let $G_m \to \Gstar$ in the bounded-Lipschitz metric $\Dbl$, which is equivalent to weak convergence since $\nuisancespace$ is a compact metric space by Assumption~\ref{assu:nuisance_parameters}. We want to show that $\TV(f(\cdot; G_m), f(\cdot; \Gstar)) \to 0$.
By Assumption~\ref{assu:nuisance_statistics}, the family of distributions of $U_i$ is tight. Thus, for any $\varepsilon > 0$, there exists a compact set $C \subset \mathcal{U}$ such that $\sup_{\nu \in \nuisancespace} \PP[\nu]{U_i \notin C} < \varepsilon$. This implies that $\int_{C^c} f(u; G) \dd\lambda(u) < \varepsilon$ for any $G$, including $G_m$ and $\Gstar$.
The total variation distance can be bounded as:
$$
\TV(f(\cdot; G_m), f(\cdot; \Gstar)) \leq \frac{1}{2}\int_C \abs{f(u; G_m) - f(u; \Gstar)} \dd\lambda(u) + \frac{\varepsilon}{2}.
$$
For any $u \in C$, the function $\nu \mapsto p(u \mid \nu)$ is bounded and continuous by Assumption~\ref{assu:nuisance_statistics}. Since $G_m \to \Gstar$ weakly, we have $f(u; G_m) = \int p(u \mid \nu) G_m(\dd\nu) \to \int p(u \mid \nu) \Gstar(\dd\nu) = f(u; \Gstar)$ for each $u \in C$.
Furthermore, the functions $\nu \mapsto p(u \mid \nu)$ for $u \in C$ are uniformly bounded by some constant $M$. Thus, $\abs{f(u; G_m) - f(u; \Gstar)} \leq 2M$ for all $u \in C$. Since $\lambda(C) < \infty$, we can apply the dominated convergence theorem to conclude that $\int_C \abs{f(u; G_m) - f(u; \Gstar)} \dd\lambda(u) \to 0$. As $\varepsilon$ was arbitrary, this establishes the continuity of the map $G \mapsto f(\cdot; G)$ from the weak topology to the total variation topology.\\

\noindent \textbf{Step 2: Separation.} Let $A_\delta := \{G \in \mathcal{P}(\nuisancespace) : \Dbl(G, \Gstar) \ge \delta\}$, where $\mathcal{P}(\nuisancespace)$ is the space of probability measures supported on $\nuisancespace$.
Since $\nuisancespace$ is compact, $\mathcal{P}(\nuisancespace)$ is compact under the weak topology. The set $A_\delta$ is a closed subset of a compact space, hence it is compact. Now consider the map
$$
T: \mathcal{P}(\nuisancespace) \to [0,1],\;\; G \mapsto \TV(f(\cdot; G), f(\cdot; \Gstar)).
$$
The map $T$ is continuous from $\mathcal{P}(\nuisancespace)$ with the weak topology to $[0,1]$, because it is a composition of the continuous map $G \mapsto f(\cdot; G)$ (from Step 1) and the continuous map $f \mapsto \TV(f, f(\cdot; \Gstar))$. 
Restricting $T$ to $A_{\delta}$, we find that there exists $G' \in A_{\delta}$ such that $\inf_{G \in A_\delta} T(G) = T(G')$. Since $G' \neq G$, by identifiability, we thus must have that $T(G') >0$. That is, there exists an $\varepsilon_\delta > 0$ such that $\inf_{G \in A_\delta} T(G) = \varepsilon_\delta > 0$.\\ 

\noindent \textbf{Step 3: Posterior consistency.} The result from Step 2 implies the inclusion of events:
$$
\{G : \Dbl(G, \Gstar) \ge \delta\} \subseteq \{G : \TV(f(\cdot; G), f(\cdot; \Gstar)) \ge \varepsilon_\delta\}.
$$
Taking posterior probabilities conditional on $U_1, \ldots, U_n$ on both sides, we get:
$$
\Pi(G : \Dbl(G, \Gstar) \ge \delta \mid U_1, \ldots, U_n) \le \Pi(G : \TV(f(\cdot; G), f(\cdot; \Gstar)) \ge \varepsilon_\delta \mid U_1, \ldots, U_n).
$$
Taking expectation with respect to the data-generating distribution $\Gstar$, the right-hand side converges to 0 as $n \to \infty$ by the proposition's assumption. This proves that the posterior for $G$ is consistent in the bounded-Lipschitz metric. Finally, posterior consistency in the $\Dbl$ metric implies that the posterior mean $G_{\Pi}[U_1,\ldots,U_n]$ converges to $\Gstar$ in $\Dbl$ in expectation, as shown in~\citet[Theorem 6.8]{ghosal2017fundamentals}. This verifies condition $(*)$ of Theorem~\ref{theo:calibration_n_to_infty_empirical_bayes}.
\end{proof}

\section{Proof of Theorem~\ref{theo:calibration_n_to_infty_frequentist}}
\begin{proof}
The proof follows a similar structure to that of Theorem~\ref{theo:calibration_n_to_infty_empirical_bayes}, but adapted to the frequentist frame. We first show that the partially Bayes p-values $\cpvalue_i$ are close in $L_1$ to certain oracle compound p-values. Then, we use Lemma~\ref{lemm:l1_to_uniformity}(b) to establish the result.

Fix $i \in \Hnull$. Let $\nuisancevectortrue = (\nuisancetrue_1, \ldots, \nuisancetrue_n)$ be the fixed vector of true nuisance parameters. Let $G(\nuisancevectortrue) := \frac{1}{n}\sum_{j=1}^n \delta_{\nuisancetrue_j}$ be the empirical distribution of the nuisance parameters. We define the oracle compound p-value as $P_i^C := \Pvalfunoracle(T_i, U_i, G(\nuisancevectortrue))$. By Proposition~\ref{prop:compound_pvalues}, $P_1^C,\ldots,P_n^C$ are compound p-values.

Our first goal is to show that:
\begin{equation}
\label{eq:l1_consistency_compound}
\frac{1}{n}\sum_{i=1}^n\EE[\nuisancevectortrue]{\abs{\cpvalue_i - P_i^C}} \to 0 \quad \text{as } n \to \infty.
\end{equation}
By Theorem~\ref{theo:oracle_representation}(b'), the partially Bayes p-value is $\cpvalue_i = \Pvalfunoracle(T_i, U_i, G_{\Pi}^{-i})$, where we use the shorthand $G_{\Pi}^{-i} := G_{\Pi}[U_1,\ldots,U_{i-1},U_{i+1}, \ldots,U_n]$.
By Assumption~\ref{assu:nuisance_statistics}, for any $\varepsilon > 0$, there exists a compact set $C \subset \mathcal{U}$ such that $\sup_{\nu \in \nuisancespace} \PP[\nu]{U_i \notin C} < \varepsilon$.

Arguing as in the proof of Theorem~\ref{theo:calibration_n_to_infty_empirical_bayes} in Supplement~\ref{subsec:proof_calibration_n_to_infty_empirical_bayes}, we have that for any $\eta > 0$, we can prove the following analogous bound to \eqref{eq:bound_on_diff_pvalues_on_C} (with $B>0$ also defined analogously):
$$
\abs{\Pvalfunoracle(T_i, U_i, G_{\Pi}^{-i}) - \Pvalfunoracle(T_i, U_i, G(\nuisancevectortrue))}\ind(U_i \in C) \leq \frac{2 B \Dbl(G_{\Pi}^{-i}, G(\nuisancevectortrue) )  + 4\eta}{f(U_i;G(\nuisancevectortrue) )}\ind(U_i \in C).
$$
The expectation is over $(U_1, \ldots, U_n)$ with $\nuisancetrue_1, \ldots, \nuisancetrue_n$ fixed. Noting that $U_i$ is independent of $\{U_j\}_{j \neq i}$, we find for $i \in \Hnull$:
\begin{align*}
&\EE[\nuisancevectortrue]{\frac{2B\Dbl(G_{\Pi}^{-i}, G(\nuisancevectortrue)) + 4\eta}{f(U_i; G(\nuisancevectortrue))}\ind(U_i \in C)} \\
&\;\;\;\;\;\;\;= \, \cb{2B\EE[\nuisancevectortrue_{n,-i}]{\Dbl(G_{\Pi}^{-i}, G(\nuisancevectortrue)) }+ 4 \eta} \cdot \EE[\nuisancetrue_i]{\frac{1}{f(U_i; G(\nuisancevectortrue))}\ind(U_i \in C)}.
\end{align*}
Let $G(\nuisancevectortrue_{n,-i}) := \frac{1}{n-1}\sum_{j \neq i} \delta_{\nuisancetrue_j}$. By the triangle inequality:
$$
\Dbl(G_{\Pi}^{-i}, G(\nuisancevectortrue)) \leq \Dbl(G_{\Pi}^{-i}, G(\nuisancevectortrue_{n,-i})) + \Dbl(G(\nuisancevectortrue_{n,-i}), G(\nuisancevectortrue)).
$$
The second term is $\leq C'/n$ for any $n \geq 2$ and another constant $C'$ since $\mathcal{V}$ is a compact metric space (for instance, we can take $C'$ to be twice the diameter of $\mathcal{V}$).  Combined with assumption $(*)$ of the theorem we find that there exists $n_0$ such that for all $n \geq n_0(\eta)$:
$$ 
\max_{i=1,\ldots,n }\EE[\nuisancevectortrue_{n,-i}]{\Dbl(G_{\Pi}^{-i}, G(\nuisancevectortrue))} \leq \frac{\eta}{B}.
$$
For the other term we argue by averaging over $i=1,\ldots,n$:
$$
\begin{aligned}
\frac{1}{n} \sum_{i \in \Hnull}\EE[\nuisancetrue_i]{\frac{\ind(U_i \in C)}{f(U_i; G(\nuisancevectortrue))}} &\leq \frac{1}{n} \sum_{i=1}^n\EE[\nuisancetrue_i]{\frac{\ind(U_i \in C)}{f(U_i; G(\nuisancevectortrue))}} \\ 
&=  \frac{1}{n} \sum_{i=1}^n \int_C \frac{1}{f(u; G(\nuisancevectortrue))} p(u \mid \nuisancetrue_i) \dd\lambda(u) \\
&= \int_C \frac{1}{n} \frac{ \sum_{i=1}^n p(u \mid \nuisancetrue_i)}{f(u; G(\nuisancevectortrue))} \dd\lambda(u)\\
 &= \lambda(C) < \infty.
\end{aligned}
$$
Note that the above argument constitutes the main difference of this proof as compared to the proof of Theorem~\ref{theo:calibration_n_to_infty_empirical_bayes}. Now, returning to an argumentation line similar to that of Theorem~\ref{theo:calibration_n_to_infty_empirical_bayes}, we can establish the following bound which is analogous to \eqref{eq:main_bound_cpvalue_difference}:

\begin{equation*}
\begin{aligned}
&\frac{1}{n}\sum_{i \in \Hnull}\EE[\nuisancevectortrue]{ \abs{\cpvaluefun_i(T_i, (U_1,\dotsc,U_n);\, \Pi)- \Pvalfunoracle(T_i, U_i, G(\nuisancevectortrue))} } \leq 6\eta \lambda(C) + \varepsilon.
\end{aligned}
\end{equation*}
The above holds for all $n \geq n_0(\eta)$. Now taking $n \to \infty$, then $\eta \to 0$, and finally $\varepsilon \to 0$, we obtain \eqref{eq:l1_consistency_compound}.
We can finally apply Lemma~\ref{lemm:l1_to_uniformity}(b) with $P_i^K = \cpvalue_i$ and $Q_i^K = P_i^C$. Since $P_1^C,\ldots,P_n^C$ are compound p-values (by Proposition~\ref{prop:compound_pvalues}), the conclusion of the theorem follows.
\end{proof}

\section{Some remarks on P\'olya trees}
\label{sec:polya_trees_details}
\subsection{Conjugate updates for P\'olya trees}
\label{subsec:conjugate_updates_polya_trees}

We first briefly recall a well-known fact about the conjugacy of P\'olya trees. Suppose
$$Z_1,\dotsc,Z_m \simiid G,\;\; G \sim \mathrm{PT}(\mathcal{A}, G_0, J).$$
Then the posterior distribution of $G$ given $Z_1,\dotsc,Z_m$ is again a {P\'olya} Tree,
$$G  \mid Z_1,\dotsc,Z_m  \sim \mathrm{PT}(\mathcal{A}(Z_1,\dotsc,Z_m), G_0, J),$$
where we define $\mathcal{A}(Z_1,\dotsc,Z_m)$ entrywise as
$$\alpha(j,\ell; Z_1,\dotsc,Z_m) := \alpha(j, \ell) + \#\cb{i: k_j(Z_i) = \ell }.$$

\subsection{Symmetrized P\'olya trees}
\label{subsec:conjugate_updates_symmetrized_polya_trees}

Here we explain how to perform conjugate updates for the symmetrized P\'olya tree prior $\mathrm{SymmPT}(\mathcal{A}_0, G_0^W, J)$. Recall that we defined this prior through the following two-step generative process for drawing $W \sim \mathrm{SymmPT}(\mathcal{A}_0, G_0^W, J)$:
\begin{enumerate}[noitemsep,leftmargin=*]
\item Draw $\widetilde{W} \sim \mathrm{PT}(\mathcal{A}_0, G_0^W, J)$.
\item Set $W(A) = \{\widetilde{W}(A) + \widetilde{W}(-A)\}/2$ for all measurable sets $A$.
\end{enumerate}
For what follows, we will assume that $G_0^W$ is a distribution supported on $\RR_{\geq 0}$. Note that in this caser, \smash{$\widetilde{W}$} is also supported on $\RR_{\geq 0}$ almost surely. Consequently, there is a bijection between \smash{$\widetilde{W}$} and $W$. 

Suppose we have $m$ iid samples:
$$
Z_1,\dotsc,Z_m \simiid W,\;\; W \sim \mathrm{SymmPT}(\mathcal{A}_0, G_0^W, J).
$$
Now write $Z_i = |Z_i| \varepsilon_i$ where $\varepsilon_i \in \{-1,1\}$ is a Rademacher random variable indicating the sign of $Z_i$ (and is a coin flip when $Z_i=0$). Then we see that,
$$
|Z_1|,\dotsc,|Z_m| \simiid \widetilde{W},\;\; \widetilde{W} \sim \mathrm{PT}(\mathcal{A}_0, G_0^W, J),
$$
and moreover, $(\varepsilon_1,\dotsc,\varepsilon_m)$ are independent of $\widetilde{W}$ conditional on $(|Z_1|,\dotsc,|Z_m|)$.
The above imply that we can compute the posterior of $W$ through the following two steps:
\begin{enumerate}[noitemsep,leftmargin=*]
\item First compute the posterior of $\widetilde{W}$ given $|Z_1|,\dotsc,|Z_m|$ using the conjugate update for P\'olya trees recalled in Supplement~\ref{subsec:conjugate_updates_polya_trees}.
\item Then, the posterior of $W$ is obtained by symmetrizing the posterior of $\widetilde{W}$.
\end{enumerate}

\section{Further details on computation}

\subsection{MCMC for normal means with unknown and varying variance}
\label{subsec:mcmc_normal_means}

We first describe our approach for computing partially Bayes p-values in the setting of Section~\ref{sec:normal_means}. Therein, posterior computation is very standard and we can employ well-known MCMC algorithms for Dirichlet process mixture models. 
In our implementation, we use a conjugate inverse-gamma base distribution $G_0 = \text{InvScaledChiSq}(\nu_0, \sigma_0^2)$ in~\eqref{eq:DP}, and so we can use a Gibbs sampler based on Neal's Algorithm 2~\citep{neal2000markov}.

In Section~\ref{subsec:computation_normal_means} we already previewed the state variables at the $b$-th iteration of the MCMC algorithm:
\begin{itemize}[noitemsep,leftmargin=*]
\item cluster assignments $c_1^{(b)}, \dotsc, c_n^{(b)}$ where $c_i^{(b)} \in \{1,\dotsc,K^{(b)}\}$ and $K^{(b)}$ is the number of occupied clusters at iteration $b$;
\item cluster variance parameters $\sigma_1^{2(b)}, \dotsc, \sigma_{K^{(b)}}^{2(b)}$;
\item concentration parameter $c^{(b)}$ of the Dirichlet process.
\end{itemize}
For the initialization $(b=0)$, we set $K^{(0)} = 1$, $c_i^{(0)} = 1$ for all $i$, $\sigma_1^{2(0)}=1$.
At each iteration, we update these variables in the following order (we omit the superscript $(b)$ for notational clarity):
\begin{enumerate}[leftmargin=*, wide]
\item \textbf{Cluster assignments.} For each $i = 1,\dotsc,n$ we proceed as follows. Let $n_{k,-i}$ denote the current number of observations in cluster $k$ excluding observation $i$. We compute cluster assignment probabilities as follows.
\begin{itemize}
\item (Existing clusters) For all $k \in \{1,...,K\}$, let $ \pi_k := n_{k,-i} \cdot p(S_i^2 \mid \sigma_k^2)$. 
\item (New cluster) Let $\pi_{K+1} := c \cdot \int p(S_i^2 \mid \sigma^2) G_0(\dd \sigma^2)$.
\end{itemize}
Then, renormalize the probabilities $\pi_k$ to sum to $1$, and sample $c_i$ from the resulting categorical distribution.\footnote{Our notation here omits bookkeeping of removing empty clusters.} Note that the marginal likelihood integral can be computed analytically due to conjugacy.
\item \textbf{Cluster parameters.} Consider the $k$-th occupied cluster. This cluster has $n_k$ assigned observations $\{S_i^2: c_i = k\}$. Since the base distribution is of the form $G_0 = \mathrm{inv}\chi^2(\hat{\nu}_0, \hat{\sigma}_0^2)$, the posterior distribution of $\sigma_k^2$ is equal to
$$ \mathrm{inv}\chi^2\left(\hat{\nu}_0 + n_k(K-1),\; \frac{\hat{\nu}_0\hat{\sigma}_0^2 + (K-1)\sum_{i: c_i=k} S_i^2}{\hat{\nu}_0 + n_k(K-1)}\right).$$
We then sample $\sigma_k^2$ from this posterior distribution.
\item \textbf{Concentration parameter.} We update $c$ using the auxiliary variable method of~\citet{escobar1995bayesian}. Recall that the prior distribution of $c$ is set as $c \sim \text{Gamma}(a,b)$, where $a = 0.001$ and $b = 100$. The update proceeds as follows.
We first sample $\eta \sim \text{Beta}(c+1, n)$, where $c$ is the current value of the concentration parameter. 
Then let $k^*$ be the current number of occupied (that is, non-empty) clusters, $b^* = (1/b - \log \eta)^{-1}$, and $w^*= (a+k^*-1)/(a+k^*-1+n/b^*)$. Finally, we sample a new $c$ from the following two-component mixture:
$$c \sim w^* \cdot \text{Gamma}(a+k^*, b^*) + (1-w^*) \cdot \text{Gamma}(a+k^*-1, b^*).$$
\end{enumerate}
We always use $5,000$ burn-in iterations. For the real data applications, we run another $100,000$ iterations from which we collect samples, while in the simulation study we run $10,000$ iterations.

\subsection{MCMC for location problems with unknown shape and scale}
\label{subsec:mcmc_polya}
We first recall the state variables that we maintain at iteration $b$ (already discussed in Section~\ref{subsec:computation_unknown_shape}).
\begin{itemize}[noitemsep,leftmargin=*]
\item cluster assignments $c_1^{(b)}, \dotsc, c_n^{(b)}$ where $c_i^{(b)} \in \{1,\dotsc,K^{(b)}\}$ and $K^{(b)}$ is the number of occupied clusters at iteration $b$;
\item cluster variance parameters $\sigma_1^{2(b)}, \dotsc, \sigma_{K^{(b)}}^{2(b)}$;
\item concentration parameter $c^{(b)}$ of the Dirichlet process;
\item symmetrized P\'olya tree realization $W^{(b)}$ from $\mathrm{SymmPT}(\mathcal{A}, G_0^W, J)$;
\item null imputed test-statistics $\bar{Z}_1^{(b)}, \dotsc, \bar{Z}_n^{(b)}$ and datasets $\dataset_i^{(b)} = \{Z_{i1}^{(b)}, \dotsc, Z_{iK}^{(b)}\}$, where $Z_{ij}^{(b)} = U_{ij} +  \bar{Z}_i^{(b)}$. 
\end{itemize}
For initialization we proceed as follows. We first run the algorithm of Section~\ref{sec:normal_means}/Supplement~\ref{subsec:mcmc_normal_means} (that models the shape of the noise distribution as normal) until completion. We then set \smash{$K^{(0)}$} to the number of occupied clusters at the end of that run, \smash{$c_i^{(0)}$} to the cluster assignments at the end of that run, and \smash{$\sigma_k^{2(0)}$} to the cluster variances at the end of that run. Similarly, we set \smash{$c^{(0)}$} to the concentration parameter at the end of that run.
Finally, we set $\bar{Z}_i^{(0)} = 0$ for all $i$ and initialize $W^{(0)}$ from its prior distribution.

In what follows, it will be convenient to also define the standardized distribution $\mathring{W}{(b)}$ which is obtained by standardizing $W^{(b)}$ to have variance $1$, that is,
\begin{equation}
\mathring{W}^{(b)}(\cdot) = W^{(b)}\p{\,\cdot\, 
\Bigg / \sqrt{ \int u^2 W^{(b)}(\dd u)}}.
\label{eq:standardized_polya_tree}
\end{equation}
and we denote its density by $\mathring{w}^{(b)}$. This step can be computed efficiently, see Supplement~\ref{sec:variance_efficient_computation} below for details.

The algorithm we propose takes the form of Gibbs sampling with several nested Metropolis-Hastings (MH) steps.
At each iteration, we update our variables in the following order and as described below.
\begin{enumerate}[leftmargin=*, wide]
\item \textbf{Cluster assignment.} For each $i = 1,\dotsc,n$, we update $c_i^{(b)}$ using Neal's Algorithm 8~\citep{neal2000markov} with auxiliary variables. The reason we use Algorithm 8 instead of Algorithm 2 (as in Supplement~\ref{subsec:mcmc_normal_means}) is that here the base distribution $G_0^W$ is not conjugate to the likelihood induced by the realized P\'olya tree.
Neal's Algorithm 8 introduces auxiliary parameters to handle the non-conjugacy.

Now fix $i \in \{1,\dotsc,n\}$. We seek to update the assignment of observation $i$. 
Algorithm 8 maintains auxiliary parameters $\phi_1, \dotsc, \phi_m$ (with $m = 10$ in our implementation) that are refreshed at each iteration, however, we omit this from our notation. These parameters are generated as follows:
$$
\phi_2,\dotsc,\phi_{m} \simiid G_0^W.
$$
Meanwhile, $\phi_1$ is set as follows. If the current cluster assignment of observation $i$ only contains observation $i$ itself, then set $\phi_1 = \sigma_{c_i}^{2(b)}$. Otherwise, indendently sample $\phi_1 \sim G_0^W$. Next, compute assignment probabilities proportional to:
\begin{itemize}
\item For auxiliary parameter $j \in \{1,\dotsc,m\}$: $\;\;\frac{c^{(b)}}{m} \phi_j^{-K/2}\displaystyle\prod_{j=1}^K \mathring{w}^{(b)}(Z_{ij}^{(b)}/ \sqrt{\phi_j}).$
\item For existing cluster $k$: $\;\;n_{k,-i}^{(b)}  \p{\sigma_k^{(b)}}^{-K}\displaystyle\prod_{j=1}^K \mathring{w}^{(b)}(Z_{ij}^{(b)}/ \sigma_k^{(b)})$.
\end{itemize}
Above, \smash{$n_{k,-i}^{(b)}$} is the number of observations in cluster $k$ excluding observation $i$.
The new cluster assignment for the $i$-th observation is then sampled from the resulting categorical distribution.
If observation $i$ is assigned to auxiliary parameter $j$, then
a new cluster is created with parameter \smash{$\sigma_{\text{new}}^{2(b)} = \phi_j^{(b)}$}. 
\item \textbf{Cluster parameters.} Consider the $k$-th occupied cluster. We update $\sigma_k^{2(b)}$ conditional on the imputed datasets and the current P\'olya tree realization with Metropolis-Hastings. We generate our proposal as follows. First, we draw a sample from the posterior that would arise if the noise were normal, i.e., precisely as described in Step 2.\ of Supplement~\ref{subsec:mcmc_normal_means}. Then we further multiply this by an independent $\chi^2_5/5$ variate to add extra variability. 
We take $3$ MH steps for each \smash{$\sigma_k^{2(b)}$} update.

\item \textbf{Null-imputed test statistics.} For each configuration sample in cluster $k$, use Metropolis-Hastings to impute \smash{$\bar{Z}_i^{(b)}$} under the null hypothesis $\mu_i = 0$ given $U_i$, \smash{$\tilde{\sigma}_k^{2(b)}$}, and \smash{$W^{(b)}$}. We have that,
$$
p(\bar{z}_i \mid \mu_i=0, U_i, \tilde{\sigma}_k^{2(b)}, W^{(b)}, c_i^{(b)}=k) \propto \prod_{j=1}^K \mathring{w}^{(b)}\p{\frac{U_{ij} + \bar{z}_i}{\sigma_k^{(b)}}}.
$$
Meanwhile, we generate our proposal from the t distribution with 5 degrees of freedom scaled by \smash{$\sigma_k^{(b)}/\sqrt{K}$}.\footnote{The motivation is that, under normality of $W_i^{(b)}$, the conditional distribution of \smash{$\bar{Z}_i^{(b)}$} would be precisely \smash{$\mathrm{N}(0, \sigma_k^{2(b)}/K)$}. Our proposal distribution is a more spread out and heavy-tailed version of this, replacing the normal with a t distribution with 5 degrees of freedom.} We take $3$ MH steps for each \smash{$\bar{Z}_i^{(b)}$} update.

\item \textbf{Concentration parameter.} We can update $c^{(b)}$ using the auxiliary variable method of~\citet{escobar1995bayesian}, exactly as described in Supplement~\ref{subsec:mcmc_normal_means}.

\item \textbf{Symmetrized P\'olya tree.} Define $m=n \cdot K$ observations as follows: 
$$\tilde{Z}_{ij}^{(b)} = Z_{ij}^{(b)}\cdot \frac{ \sqrt{ \int u^2 W^{(b)}(\dd u)}}{\sigma_{c_i^{(b)}}^{(b)}},\;\;i=1,\dotsc,n,\;j=1,\dotsc,K.$$
Then we conduct the conjugate update of the symmetrized P\'olya tree given the data \smash{$\{\tilde{Z}_{ij}^{(b)}: i=1,\dotsc,n, j=1,\dotsc,K\}$} as described in Supplement~\ref{subsec:conjugate_updates_symmetrized_polya_trees}.
\end{enumerate}

After the initialization, we use further $2,000$ burn-in iterations. Afterwards, for the real data applications, we run $100,000$ iterations, while in the simulation study, we run $10,000$ iterations.

\subsubsection{Computation of variance for truncated symmetrized P\'olya trees}
\label{sec:variance_efficient_computation}

Let $W$ be a truncated (symmetrized) P\'olya tree with base distribution $G_0^W$. Throughout our MCMC algorithm, we need to compute the second moment of $W$ multiple times, see~\eqref{eq:standardized_polya_tree}. Thus it is important to do this efficiently.

The basic idea is as follows. Write $w(\cdot)$ for the density of $W$ and $g_0^W(\cdot)$ for the density of $G_0^W$. Also write $x_0,\ldots,x_{2^J+1}$ for the $0,2^{-J},\ldots, 1$-quantiles of $G_0^W$. In our applications, typically, $x_0=0$ and $x_{2^J+1} = \infty$ (e.g., this is the case for $G_0^W = |\mathring{t}_8|$.)

We can then write $w(\cdot)$ in the following form:
$$w(x) = \sum_{\ell=1}^{2^J+1} \ind\cb{x \in [x_{\ell-1}, x_\ell)} p_{\ell} g_0^W(x),$$
for some numbers $p_{\ell} \geq 0$ with $\sum_{\ell} p_{\ell} = 2^J$ that we can compute by traversing the binary splits of the P\'olya tree.
This representation implies that,
$$
\int x^2 W(\dd x) = \sum_{\ell=1}^{2^J+1} p_{\ell} \int_{x_{\ell-1}}^{x_\ell} x^2 g_0^W(x) dx.
$$
The upshot is that we can compute the latter integrals analytically for appropriate choices of $G_0^W$. 

As an example, let $f_{t,8}(\cdot)$ be the density of the t distribution with $8$ degrees of freedom. Then,
$$
\int_a^b x^2 f_{t,8}(x) \dd x =
U(b)- U(a),\;\;U(t) := \frac{2 t^3 (t^4 + 28 t^2 + 280)}{3 (t^2 + 8)^{7/2}},\;\; U(\pm \infty) := \pm 2/3.
$$
From the above, we can directly compute \smash{$\int_{x_{\ell-1}}^{x_\ell} x^2 g_0^W(x) dx$} when $G_0^w$ is the $|\mathring{t}_8|$ distribution (using the fact that \smash{$|\mathring{t}_8|$} is the folded density of the standardized $t_8$ distribution).

\end{document}